\documentclass[10pt, journal, letterpaper]{IEEEtran}

\usepackage{amsmath,amsthm,amssymb,dsfont,bm,nicefrac}
\usepackage[mathscr]{eucal}
\usepackage[inline]{enumitem}
\usepackage{makecell}
\usepackage{color,verbatim,comment,bbm}
\usepackage{algorithm,algpseudocode}
\usepackage{epstopdf}
\usepackage{mathtools}

\usepackage{url}
 \PassOptionsToPackage{hyphens}{url}
 \usepackage[colorlinks,
             linkcolor=blue,
             citecolor=blue,
             urlcolor=magenta,
             pdftex
 ]{hyperref}

\newtheorem{theorem}{Theorem}
\newtheorem{lemma}{Lemma}

\newtheorem{definition}{Definition}
\newtheorem{remark}{Remark}

\newcommand{\n}{\nonumber}
\newcommand\numberthis{\addtocounter{equation}{1}\tag{\theequation}}
\newcounter{const-no}

\DeclarePairedDelimiterX{\card}[1]{\lvert}{\rvert}{#1}
\DeclarePairedDelimiterX{\norm}[1]{\lVert}{\rVert}{#1}

\DeclareMathOperator*{\argmax}{arg\,max}

\newcommand*\diff{\mathop{}\!\mathsf{d}}
\newcommand{\ind}[1]{\mathds{1}\left\lbrace #1 \right\rbrace}
\newcommand\givenbase[1][]{\:#1\lvert\:}
\let\given\givenbase

\DeclarePairedDelimiterX\Basics[1](){\let\given\sgiven #1}

\def\EE{{\mathbb{E}}}\def\PP{{\mathbb{P}}}
\def\RR{\mathbb{R}}\def\NN{\mathbb{N}}

\begin{document}

\title{Augmenting Max-Weight with Explicit Learning for Wireless Scheduling with Switching Costs}
\author{Subhashini~Krishnasamy,
        Akhil~P~T,
        Ari~Arapostathis,~\IEEEmembership{Fellow,~IEEE,}
        Rajesh~Sundaresan,~\IEEEmembership{Senior~Member,~IEEE,}
        and~Sanjay~Shakkottai,~\IEEEmembership{Fellow,~IEEE}
\thanks{S. Krishnasamy, A. Arapostathis and S. Shakkottai are with the Department
of Electrical and Computer Engineering, The University of Texas at
Austin, Austin, TX 78712 USA (e-mail: subhashini.kb@utexas.edu).}
\thanks{P. T. Akhil is with the Department of Electrical Communication Engineering, Indian Institute of Science, Bangalore 560012, India.}
\thanks{R. Sundaresan is with the Department of Electrical Communication Engineering and with The Robert Bosch Centre for Cyber-Physical Systems, Indian Institute of Science, Bangalore 560012, India.}
\thanks{A shorter version of this paper appears in the
  Proceedings of IEEE Conference on Computer Communications (IEEE
  Infocom 2017) \cite{KAASS17}.}}
%

%

\maketitle
\begin{abstract}
  In small-cell wireless networks where users are connected to multiple base stations (BSs), it is often advantageous to switch off dynamically a subset of BSs to minimize energy costs. We consider two types of energy cost: (i) the cost of maintaining a BS in the active state, and (ii) the cost of switching a BS from the active state to inactive state. The problem is to operate the network at the lowest possible energy cost (sum of activation and switching costs) subject to queue stability. In this setting, the traditional approach --- a Max-Weight algorithm along with a Lyapunov-based stability argument --- does not suffice to show queue stability, essentially due to the temporal co-evolution between channel scheduling and the BS activation decisions induced by the switching cost. Instead, we develop a learning and BS activation algorithm with slow temporal dynamics, and a Max-Weight based channel scheduler that has fast temporal dynamics. We show using convergence of time-inhomogeneous Markov chains, that the co-evolving dynamics of learning, BS activation and queue lengths lead to near optimal average energy costs along with queue stability.
\end{abstract}


 \section{Introduction}
\label{sec:intro}
Due to the tremendous increase in demand for data traffic, modern cellular
networks have taken the densification route to support peak traffic
demand \cite{nagabhushan-etal14nw-densification-5g}. While increasing
the density of base-stations gives greater spectral efficiency, it
also results in increased costs of operating and maintaining the
deployed base-stations. Rising energy cost is a cause for concern, not
only from an environmental perspective, but also from an economic
perspective for network operators as it constitutes a significant
portion of the operational expenditure. To address this challenge,
latest research aims to design energy efficient networks that balance
the trade-off between spectral efficiency, energy efficiency and user
QoS requirements \cite{wu2015recent, oh2011toward}.

Studies reveal that base-stations contribute to more than half of the
energy consumption in cellular networks \cite{marsan2009optimal,
  wu2015energy}. Although dense deployment of base-stations are useful
in meeting demand in peak traffic hours, they regularly have excess
capacity during off-peak hours \cite{oh2011toward,jie2012dynamic}. A fruitful way to
conserve power is, therefore, to dynamically switch off under-utilized
base-stations. Even in networks that do not have fluctuations in traffic load, switching base-stations dynamically is a useful way to reduce power consumption while meeting the network traffic demand. For this purpose, modern cellular standards incorporate
protocols that include \emph{sleep} and \emph{active} modes for
base-stations. The sleep mode allows for selectively switching
under-utilized base-stations to low energy consumption modes. This
includes completely switching off base-stations or switching off only
certain components.

Consider a time-slotted multi base-station (BS) cellular network where
subsets of BSs can be dynamically activated. Since turning off BSs
could adversely impact the performance perceived by users, it is
important to consider the underlying energy vs. performance trade-off
in designing BS activation policies. In this paper, we study the joint
problem of dynamically selecting the BS activation sets and user
rate allocation depending on the network load. We take into account
two types of overheads involved in implementing different activation
modes in the BSs.

\noindent {\bf (i) Activation cost} occurs due to maintaining a BS in
the active state. This includes energy spent on main power supply, air
conditioning, transceivers and signal processing
\cite{jie2012dynamic}. Surveys show that a dominant part of the energy
consumption of an active base-station is due to static factors that do
not have dependencies with traffic load intensities \cite{oh2011toward,
  arnold2010power}. Therefore, an active BS consumes almost the same
energy irrespective of the amount of traffic it serves. Typically, the
operation cost (including energy consumption) in the sleep state is
much lower than that in the active state since it requires only
minimal maintenance signaling \cite{wu2015energy}.

  \noindent {\bf (ii) Switching cost} is the penalty due to switching
  a BS from active state to sleep state or vice-versa. This factors in
  the signaling overhead (control signaling to users,
  signaling over the backhaul to other BSs and/or the BS controller),
  state-migration processing, and switching energy consumption
  associated with dynamically changing the BS modes
  \cite{jie2012dynamic}.

  Further, switching between these states typically cannot occur
  instantaneously. Due to the hysteresis time involved in migrating
  between the active and sleep states, BS switching can be done only
  at a slower time-scale than that of channel scheduling
  \cite{abbasi2013distributed, zheng2015optimal}.


\subsection*{Main Contributions}
\label{sec:contrib}

We formulate the problem in a (stochastic) network cost minimization
framework. The task is to select the set of active BSs in every
time-slot, and then based on the instantaneous channel state for the
activated BSs, choose a feasible allocation of rates to users. Our
aim is to minimize the total network cost (sum of activation and
switching costs) subject to stability of the user queues at the BSs.

While BS switching can be used to reduce energy costs both when the traffic load is dynamic and static, we consider the static case in this paper. Specifically, we assume that the incoming traffic for each user to a BS is independent and identically distributed (i.i.d.) with fixed rates. In this stationary setting, the task is to find the right way to activate and de-activate BSs so as to serve the incoming load while minimizing the energy cost. This is challenging especially because the energy cost includes the cost of switching the BSs from one state to the other. In practice, one could model the non-stationary setting as one with {\em regime changes}. One could then separately apply the main findings of our i.i.d. traffic load study to each regime. Our simulation studies described later in the paper suggest the modifications needed for application of our findings to settings with regime changes.

{\bf Insufficiency of the standard Lyapunov technique:} Such
stochastic network resource allocation problems typically adopt greedy
primal dual algorithms along with virtual-queues to accommodate
resource constraints \cite{GeoNeeTas_06,LinShrSri_06,sryi14}. To
ensure stability, this technique crucially relies on achieving
negative Lyapunov drift in some fixed number of time-slots. In our problem, unlike in the traditional setting, such an approach cannot be applied because the rates available for allocation in a time-slot is correlated with the network state in the previous time-slot.
See Section \ref{optFramework-Discussion-VirtualQueues}.1 for more details.

To circumvent difficulties introduced through this co-evolution, we
propose an approach that uses queue-lengths for channel scheduling at
a fast time-scale, but explicitly uses arrival and channel statistics
(using learning via an explore-exploit learning policy) for activation
set scheduling at a slower time-scale. Our main contributions are as
follows.

%
%
\begin{enumerate}

\item {\bf Static-split Activation + Max-Weight Channel
    Scheduling:}  We propose a solution that explicitly controls the time-scale separation between BS activation and rate allocation decisions. At BS switching instants (which occurs at a slow time-scale), the strategy uses
  a static-split rule (time-sharing) which is pre-computed using the explicit knowledge of the arrival and channel statistics for selecting the activation state. This activation
  algorithm is combined with a  queue-length based
  Max-Weight algorithm for rate allocation (applied at the fast
  time-scale of channel variability). We show that the joint dynamics of these two algorithms lead to stability; further, the choice of parameters for the algorithm enables us to achieve an average network cost that is arbitrarily close to the optimal cost.

\item {\bf Learning algorithm with provable guarantees:} In the
  setting where the arrival and channel statistics are not known, we
  propose an \textit{explore-exploit} policy that estimates arrival
  and channel statistics in the explore phase, and uses the estimated
  statistics for activation decisions in the exploit phase (this phase includes BS switching at a slow
  time-scale). This is combined with a
  Max-Weight based rate allocation rule restricted to the activated BSs (at a fast time-scale).  We prove that this joint
  learning-cum-scheduling algorithm can ensure queue stability while
  achieving close to optimal network cost.

\item {\bf Convergence bounds for time-inhomogeneous Markov chains:}
  In the course of proving the theoretical guarantees for our
  algorithm, we derive useful technical results on convergence of
  time-inhomogeneous Markov chains. More specifically, we derive
  explicit convergence bounds for the marginal distribution of a
  finite-state time-inhomogeneous Markov chain whose transition
  probability matrices at each time-step are arbitrary (but small)
  perturbations of a given stochastic matrix. We believe that these
  bounds are useful not only in this specific problem, but are of
  independent interest.
\end{enumerate}

To summarize then, our approach can be viewed as an algorithmically
engineered separation of time-scales for only the activation set
dynamics, while adapting to the channel variability for the queue
dynamics. Such an engineering of time-scales leads to coupled
fast-slow dynamics, the `fast' due to opportunistic channel allocation
and packet queue evolution with Max-Weight, and the `slow' due to
infrequent base-station switching using learned statistics. Through a
novel Lyapunov technique for convergent time-inhomogeneous Markov
chains, we show that we can achieve queue stability while operating at
a near-optimal network cost.

\subsection*{Related Work}
While mobile networks have been traditionally designed with the
objective of optimizing spectral efficiency, design of energy
efficient networks has been of recent interest. A survey of various
techniques proposed to reduce operational costs and carbon footprint
can be found in \cite{oh2011toward, hasan2011green, wu2015recent,
  wu2015energy}. The survey in \cite{wu2015energy} specially focuses
on sleep mode techniques in BSs.

Various techniques have been proposed to exploit BS sleep mode to
reduce energy consumption in different settings. Most of them aim to
minimize energy consumption while guaranteeing minimum user QoS
requirements. For example, \cite{han2012energy, jie2012dynamic,
  gong2014base} consider inelastic traffic and consider outage
probability or blocking probability as metrics for measuring QoS.  In
\cite{abbasi2013distributed}, the problem is formulated as a utility
optimization problem with the constraint that the minimum rate demand
should be satisfied. But they do not explicitly evaluate the
performance of their algorithm with respect to user QoS.  The authors
in \cite{kamitsos2010optimal, guo2016delay} model a single BS scenario
with elastic traffic as an M/G/1 vacation queue and characterize the
impact of sleeping on mean user delay and energy consumption. In
\cite{zheng2015optimal}, the authors consider the multi BS setting
with Poisson arrivals and delay constraint at each BS.

Most papers that study BS switching use models that ignore switching
cost. Nonetheless, a few papers acknowledge the importance of avoiding
frequent switching. For example, Oh et al. \cite{oh2013dynamic}
implement a hysteresis time for switching in their algorithm although
they do not consider it in their theoretical analysis.  Gou et
al. \cite{guo2016delay} also study hysteresis sleeping schemes which
enforce a minimum sleeping time. In \cite{abbasi2013distributed} and
\cite{zheng2015optimal}, it is ensured that interval between switching
times are large enough to avoid overhead due to transient network
states. Finally Jie et al. \cite{jie2012dynamic} consider BS sleeping
strategies which explicitly incorporate switching cost in the model
(but they do not consider packet queue dynamics). They emphasize that
frequent switching should be avoided considering its effect on
signaling overhead, device lifetime and switching energy consumption,
and also note that incorporating switching cost introduces time
correlation in the system dynamics.

Finally, this paper builds on the rich MaxWeight literature for
opportunistic scheduling
\cite{taseph92,AndKumRamStoVijWhi_00,neemodroh05}. The literature has
considered many aspects of utility maximization and tail performance
\cite{stolyar2005maximizing,LinShrSri_06,venlin10}, partial channel
information \cite{NeRaLaP12,gocash12}, and heterogeneous and
inconsistent network information \cite{yish11}; we refer to
\cite{sryi14} for a comprehensive survey. Most related among these are
the studies with partial information and two-stage decision making
\cite{NeRaLaP12,gocash12,manikandan2009cross}, with MaxWeight averaged through an
appropriate conditional expectation for first-stage decision making,
and the usual MaxWeight rule for the second stage, and with the proofs
of stability shown using a Lyapunov argument. Our work differs in
that the switching stemming from base-station activation does not
directly permit a standard Lyapunov argument to hold (see
Section~\ref{optFramework-Discussion-VirtualQueues}.1 for additional discussion); thus we use explicit
learning in the first stage, followed by the usual MaxWeight for the
second stage. Our proof technique also substantially differs, as our
first stage arguments are based on an analysis of time-inhomogeneous
Markov Chains. 



\paragraph*{Notation} Important notation for the problem setting can
be found in Table~\ref{tab:notation}.
For any two vectors
$\mathbf{v}_1$, $\mathbf{v}_2$ and scalar $a$,
$\mathbf{v}_1 \cdot \mathbf{v}_2$ denotes the dot product between the
two vectors and $\mathbf{v}_1 + a = \mathbf{v}_1 + a\mathbf{1}$.

 \section{System Model}
We consider a time-slotted cellular network with $n$ users and $M$ base-stations (BS) indexed by $u = 1, \dots, n$ and $m = 1, \dots, M$ respectively. Users can possibly be connected to multiple BSs. It is assumed that the user-BS association does not vary with time. 

\subsection{Arrival and Channel Model}
\label{subsec:arrival-channel}
Data packets destined for a user $u$ arrive at a connected BS $m$ as a bounded (at most $\bar{\mathsf{A}}$ packets in any time-slot), i.i.d.\ process $\left\lbrace A_{m,u}(t) \right\rbrace_{t \geq 1}$  with rate $\EE\left[ A_{m,u}(t) \right] = \lambda_{m,u}$. Arrivals get queued if they are not immediately transmitted. Let $Q_{m,u}(t)$ represent the queue-length of user $u$ at BS $m$ at the beginning of time-slot $t$.

The channel between the BSs and their associated users is also time-varying and i.i.d across time (but can be correlated across links), which we represent by the network channel-state process $\left\lbrace H(t) \right\rbrace_{t > 0}$. At any time $t$, $H(t)$ can take values from a finite set $\mathcal{H}$ with probability mass function given by $\bm{\mu}$. Let $\bar{\mathsf{R}}$ be the maximum number of packets that can be transmitted over any link in a single time-slot.  
We consider an abstract model for interference by working with the set $\mathcal{R}(\mathbf{1}, h) \subset \{0, 1, \dots, \bar{R}\}^{M \times n}$ defined as the set of all possible rate vectors (the number of packets that can be transmitted in a time-slot) achievable by non-randomized scheduling rules in a single time-slot, given that the channel state in that time-slot is $h$. Since the number of packets that can be transmitted per link is upper bounded by  $\bar{\mathsf{R}}$, $\mathcal{R}(\mathbf{1}, h)$ has finite cardinality. For concrete examples of interference models, we refer the reader to \cite[Ch.~2]{GeoNeeTas_06}.

\subsection{Resource Allocation}
\label{subsec:allocation}
At any time-slot $t$, the scheduler has to make two types of allocation decisions:

\noindent {\bf BS Activation:}  Each BS can be scheduled to be in one of the two states, \textit{ON} (\textit{active} mode) and \textit{OFF} (\textit{sleep} mode). Packet transmissions can be scheduled only from BSs in the \textit{ON} state. The cost of switching a BS from \textit{ON} in the previous time-slot to \textit{OFF} in the current time-slot is given by $\mathsf{C}_0$ and the cost of maintaining a BS in the \textit{ON} state in the current time-slot is given by $\mathsf{C}_1$.
The activation state at time $t$ is denoted by $\mathbf{J}(t) = \left( J_m(t) \right)_{m \in [M]}$, where $J_m(t) := \mathds{1}\{ \text{BS } m \text{ is \textit{ON} at time } t\}$. We also denote the set of all possible activation states, $\{0, 1\}^M$,  by $\mathcal{J}$. The total cost of operation, which we refer to as the \textit{network cost}, at time $t$ is the sum of switching and activation cost and is given by
\begin{align}
\label{eq:cost-def}
C(t) := \mathsf{C}_0 \norm{\left( \mathbf{J}(t-1) - \mathbf{J}(t) \right)^+}_1 + \mathsf{C}_1 \norm{\mathbf{J}(t)}_1.
\end{align}
It is assumed that the current network channel-state $H(t)$ is
unavailable to the scheduler at the time of making activation
decisions.

\noindent {\bf Rate Allocation:} The network channel-state is observed after the BSs are switched \textit{ON} and before the packets are scheduled for transmission. Moreover, only the part of the channel state restricted to the activated BSs, which we denote by $H(t)|_{\mathbf{J}(t)}$, can be observed. For any $j \in \mathcal{J}, h \in \mathcal{H}$, let $\mathcal{R}(j,h) \subset \{0, 1, \dots, \bar{R}\}^{M \times n}$ denote the set of all possible service rate vectors that can be allocated when the activation set is $j$ and the channel state is $h$. A more precise definition of $\mathcal{R}(j,h)$ is as follows. For any $j \in \mathcal{J}$, $\mathbf{r} \in \RR^{M \times n}$, let the product $\mathbf{r} \circ j$ be an $\RR^{M \times n}$ matrix defined as
\begin{align}
\label{eq:rate-mx-reduction}
(\mathbf{r} \circ j)_{m,u} = \begin{cases} r_{m,u} & \text{if } j_m = 1, \\ 0 & \text{otherwise.} \end{cases}
\end{align}
Also for any set $\mathcal{R} \subset \RR^{M \times n}$, define  $\mathcal{R} \circ j := \left\lbrace \mathbf{r} \circ j : \mathbf{r} \in \mathcal{R}  \right\rbrace$.
We assume that
\begin{enumerate*}[label=(\roman*)]
\item a BS that is merely switched \textit{ON} but not transmitting packets does not cause any interference in the network, and
\item $\mathcal{R}(\mathbf{1}, h) \circ j \subseteq \mathcal{R}(\mathbf{1}, h)$ for any $j \in \mathcal{J}$.
\end{enumerate*}
 Based on these assumptions, we define $\mathcal{R}(j,h) := \mathcal{R}(\mathbf{1}, h) \circ j$. This means that $\mathcal{R}(j',h) \subseteq \mathcal{R}(j,h)$ for any $j',j \in \mathcal{J}$ such that $j' \leq j$, and $\mathcal{R}(\mathbf{1}, h)$ contains all possible rate vectors when the channel state is $h$ for any BS activation set. Given the channel observation $H(t)|_{\mathbf{J}(t)}$, the scheduler allocates a rate vector $\mathbf{S}(t) = \left( S_{m,u}(t) \right)_{m \in [M], u \in [n]}$ from the set $\mathcal{R}(\mathbf{J}(t), H(t))$ for packet transmission. This allows for draining of $S_{m,u}(t)$ packets from user $u$'s queue at BS $m$ for all $u \in [n]$ and $m \in [M]$.

Thus the resource allocation decision in any time-slot $t$ is given by the tuple $(\mathbf{J}(t), \mathbf{S}(t))$.
The sequence of operations in any time-slot can, thus, be summarized
as follows: (i) Arrivals, (ii) BS Activation-Deactivation, (iii)
Channel Observation, (iv) Rate Allocation, and (v) Packet Transmissions.


\begin{table}
\renewcommand{\arraystretch}{}
\caption{General Notation}
\label{tab:notation}
\centering
\begin{tabular}{||c | c||}
\hline
\bfseries Symbol & \bfseries Description	\\
\hline\hline
$n$ & Number of users	\\ \hline
$M$ & Number of BSs \\ \hline
$[l]$ & The set $\{1,2,\ldots, l\}$ for an integer $l$. \\ \hline
$A_{m,u}(t)$ & Arrival for user $u$ at BS $m$ at time $t$ \\ \hline
& Maximum number of arrivals \\
$\bar{\mathsf{A}}$ & to any queue in a time-slot \\ \hline
$\bm{\lambda}$ & Average arrival rate vector \\ \hline
$H(t)$ & Channel state at time $t$ \\ \hline
$\mathcal{H}$ & Set of all possible channel states \\ \hline
$\bm{\mu}$ & Probability mass function of channel state \\ \hline
&  Maximum service rate \\
$\bar{\mathsf{R}}$ & to any queue in a time-slot \\ \hline
$h|_{j}$ & Channel state $h$ restricted to the activated BSs in $j$ \\ \hline
$\mathcal{R}(j,h) \subseteq \mathbb{R}^{M \times n}$ & Set of all possible rate vectors for \\ & activation vector $j$ and  channel state $h$ \\ \hline
$\mathbf{J}(t) = \left( J_m(t) \right)$ & Activation vector at time $t$ \\ \hline
$\mathcal{J}$ & Set of all possible activation states \\ \hline
$\mathbf{S}(t) = \left( S_{m,u}(t) \right)$ & Rate allocation at time $t$ \\ \hline
$\mathsf{C}_1$ & Cost of operating a BS in $ON$ state   \\ \hline
$\mathsf{C}_0$ & Cost of switching a BS from $ON$ to $OFF$ state \\ \hline
$C(t)$ & Network cost at time $t$ \\ \hline
$Q_{m,u}(t)$ &  Queue of user $u$ at BS $m$ \\ & at the beginning of time-slot $t$	\\ \hline
$\mathcal{P}_l$ & Set of all probability (row) vectors in $\RR^l$	\\ \hline
$\mathcal{P}^2_{l}$ & Set of all stochastic matrices in $\RR^{l \times l}$	\\ \hline
$\mathcal{W}_{l}$ & Set of all stochastic matrices in \\ & $\RR^{l \times l}$ with a single ergodic class \\ \hline
$\mathbf{1}_l$ & All $1$'s Column vector of size $l$  \\ \hline
$\mathbf{I}_l$ & Identity matrix of size $l$ \\ \hline
\hline
\end{tabular}
\end{table}

\subsection{Model Extensions}
Some of the assumptions in the model above are made for ease of
exposition and can be extended in the following ways:

\noindent {\bf (i) Network Cost}: We assume that the cost of operating a BS in
  the \textit{OFF} state (sleep mode) is zero. However, it is easy to
  include an additional parameter, say $\mathsf{C}_1'$, which denotes
  the cost of a BS in the \textit{OFF} state. Similarly, for switching
  cost, although we consider only the cost of switching a BS from
  \textit{ON} to \textit{OFF} state, we can also include the cost of
  switching from \textit{OFF} to \textit{ON} state (say
  $\mathsf{C}_0'$). The analysis in this chapter can then be extended by
  defining the network cost as
\begin{align*}
C(t) & = \mathsf{C}_0 \norm{\left( \mathbf{J}(t-1) - \mathbf{J}(t) \right)^+}_1 + \mathsf{C}_1 \norm{\mathbf{J}(t)}_1	\\
& \quad + \mathsf{C}_0' \norm{\left( \mathbf{J}(t) - \mathbf{J}(t-1) \right)^+}_1 + \mathsf{C}_1' \left( M - \norm{\mathbf{J}(t)}_1 \right)
\end{align*}
instead of \eqref{eq:cost-def}.


\noindent {\bf (ii) Switching Hysteresis Time}: While our system allows
  switching decisions in every time-slot, we will see that the key to
  our approach is a slowing of activation set switching
  dynamics. Specifically, on average our algorithm switches activation
  states once every $1/\epsilon_s$ timeslots, where $\epsilon_s$ is a
  tunable parameter. Additionally, it is easy to incorporate ``hard
  constraints'' on the hysteresis time by restricting the frequency of
  switching decisions to, say once in every $L$ time-slots (for some
  constant $L$). This avoids the problem of switching too frequently
  and gives a method to implement time-scale separation between the
  channel allocation decisions and BS activation decisions. While our
  current algorithm has inter-switching times i.i.d.\ geometric with
  mean  $1/\epsilon_s$, it is easy to allow other distributions that have
  bounded means with some independence conditions (independent of each
  other and also the arrivals and the channel). We skip details in the
  proofs for notational clarity.


 \section{Optimization Framework}
For any $t \in \NN$, let $\mathcal{F}_t = \left( \mathbf{A}(l), \mathbf{J}(l), H(l)|_{\mathbf{J}(l)}, \mathbf{S}(l) \right)_{l = 1}^{t-1}$. A policy is given by a (possibly random) sequence of resource allocation decisions $\left( \mathbf{J}(t), \mathbf{S}(t) \right)_{t >  0}$ where, at any time $t$, the decision may depend on the information from random variables observed in the past but not the future, i.e., BS activation may depend on $\mathcal{F}_t$ and rate allocation on $\left( \mathcal{F}_t, \mathbf{J}(t), H(t)|_{\mathbf{J}(t)} \right)$. Let $\left( \mathbf{J}(t-1), \mathbf{Q}(t) \right)$ be the
network state at time $t$. The rationale behind this choice of network state is to construct policies that provide control over switching costs.
\paragraph*{Notation} We use $\PP_{\varphi}\left[ \cdot \right]$ and
$\EE_{\varphi}\left[ \cdot \right]$ to denote probabilities and
expectation under policy $\varphi$. We skip the subscript when the
policy is clear from the context.

\subsection{Stability, Network Cost, and the Optimization Problem}
\begin{definition}[Stability]
\label{def:stability}
A network is said to be \emph{stable} under a policy $\varphi$ if there exist constants $\bar{Q}$, $\rho > 0$ such that for any initial condition $\left( \mathbf{J}(0), \mathbf{Q}(1) \right)$,
\begin{align}
\label{eqn:defn1}
\liminf_{T \to \infty} \frac{1}{T} \sum_{t=1}^{T}\PP_{\varphi}\left[ \sum_{m \in [m], u \in [n]} \hspace*{-.1in} Q_{m,u}(t) \leq \bar{Q} \given[\Bigg]  \mathbf{J}(0), \mathbf{Q}(1) \right] > \rho.
\end{align}
\end{definition}

\begin{remark}
The above definition of stability is applicable for a general network state process that is not necessarily Markov. It is motivated by the fact that for an aperiodic and irreducible DTMC, Definition~\ref{def:stability} implies positive recurrence. Indeed, for such a DTMC, we can conclude from (\ref{eqn:defn1}) that
\begin{align}
\label{eqn:defn1-outcome}
\limsup_{T \to \infty} \PP_{\varphi}\left[ \sum_{m \in [m], u \in [n]} \hspace*{-.1in} Q_{m,u}(t) \leq \bar{Q} \given[\Bigg]  \mathbf{J}(0), \mathbf{Q}(1) \right] > \rho
\end{align}
holds and hence the DTMC is recurrent; further (\ref{eqn:defn1-outcome}) violates the necessary condition for null recurrence (\cite[Th.~21.17]{levin2009markov}): $$\lim_{t \to \infty} \PP_{\varphi}\left[ {\bf Q}(t) = {\bf q} \given[\Bigg]  \mathbf{J}(0), \mathbf{Q}(1) \right] = 0, ~\forall {\bf q},$$
and hence the DTMC is positive recurrent.
\end{remark}

Consider the set of all ergodic Markov policies $\mathfrak{M}$, including those that know the arrival and channel statistics. A policy $\varphi \in \mathfrak{M}$ if and only if it makes (possibly randomized) allocation decisions at time $t$ based only on the current state $\left( \mathbf{J}(t-1), \mathbf{Q}(t) \right)$ (and possibly the arrival and channel statistical parameters), and the resulting network state process is an ergodic Markov chain.
%
Later, in Section \ref{subsec:discussion-optimality}, we discuss why it is sufficient to restrict attention to this class of policies. We now define the support region of a policy and the capacity region.
\begin{definition}[Support Region of a Policy $\varphi$]
\label{def:support}
The support region $\Lambda^{\varphi}(\bm{\mu})$ of a policy $\varphi$ is the set of all arrival rate vectors for which the network is stable under the policy $\varphi$.
\end{definition}
\begin{definition}[Capacity Region]
\label{def:capacity}
The capacity region $\Lambda(\bm{\mu})$ is the set of all arrival rate vectors for which the network is stable under some policy in $\mathfrak{M}$, i.e.,
$\Lambda(\bm{\mu}) := \bigcup_{\varphi \in \mathfrak{M}} \Lambda^{\varphi}(\bm{\mu}).$
\end{definition}
\begin{definition}[Network Cost of a Policy $\varphi$]
\label{def:energy-cost}
The network cost $C^{\varphi}(\bm{\mu}, \bm{\lambda})$ under a policy $\varphi$ is the long term average network cost (BS switching and activation costs) per time-slot, i.e.,
\begin{align*}
C^{\varphi}(\bm{\mu}, \bm{\lambda}) := \limsup_{T \to \infty} \frac{1}{T} \sum_{t=1}^T \EE_{\varphi}\left[ C(t) \given[\big]  \mathbf{J}(0), \mathbf{Q}(1) \right].
\end{align*}
\end{definition}

We formulate the resource allocation problem in a network cost minimization framework. Consider the problem of network cost minimization under Markov policies $\mathfrak{M}$ subject to stability. The optimal network cost is given by
\begin{align}
C^{\mathfrak{M}}(\bm{\mu}, \bm{\lambda}) := \inf_{\{\varphi \in \mathfrak{M}: \lambda \in \Lambda^{\varphi}(\bm{\mu})\}} C^{\varphi}(\bm{\mu}, \bm{\lambda}). \label{opt-problem}
\end{align}

\subsection{Markov-Static-Split Rules}
\label{subsec:ms-rules}
The capacity region $\Lambda (\bm{\mu})$ will naturally be characterized by only those Markov policies that maintain all the BSs active in all the time-slots, i.e., $\mathbf{J}(t) = \mathbf{1} \, \forall t$. In the traditional scheduling problem without BS switching, it is well-known that the capacity region can be characterized by the class of \emph{static-split} policies \cite{AndKumRamStoVijWhi_00} that allocate rates in a random i.i.d.\ fashion given the current channel state. An arrival rate vector $\bm{\lambda} \in \Lambda (\bm{\mu})$ iff there exists convex combinations $\left\lbrace \bm{\alpha}(\mathbf{1}, h) \in \mathcal{P}_{\card{\mathcal{R}(\mathbf{1}, h)}} \right\rbrace_{h \in \mathcal{H}}$ such that
\begin{align*}
\bm{\lambda} < \sum_{h \in \mathcal{H}} \mu(h) \sum_{\mathbf{r} \in \mathcal{R}(\mathbf{1},h)} \alpha_{\mathbf{r}}(\mathbf{1}, h) \mathbf{r}.
\end{align*}
But note that static-split rules in the above class, in which BSs are not switched \textit{OFF}, do not optimize the network cost.

We now describe a class of activation policies called the \emph{Markov-static-split + static-split} rules which are useful in handling the network cost. A policy is a Markov-static-split + static-split rule if it uses a time-homogeneous Markov rule for BS activation in every time-slot, and an i.i.d.\ static-split rule for rate allocations. For any $l \in \NN$, let $\mathcal{W}_l$ denote the set of all stochastic matrices of size $l$ with a single ergodic class. A Markov-static-split + static-split rule is characterized by
\begin{enumerate}
\item a stochastic matrix $\mathbf{P} \in \mathcal{W}_{\card{\mathcal{J}}}$ with a single ergodic class,
\item convex combinations $\left\lbrace \bm{\alpha}(j,h) \in \mathcal{P}_{\card{\mathcal{R}(j,h)}} \right\rbrace_{j \in \mathcal{J},h \in \mathcal{H}}$.
\end{enumerate}
Here $\mathbf{P}$ represents the transition probability matrix that specifies the jump probabilities from one activation state to another in successive time-slots. 
 $\left\lbrace \bm{\alpha}(j,h) \right\rbrace_{j \in \mathcal{J},h \in \mathcal{H}}$ specify the static-split rate allocation policy given the activation state and the network channel-state.

Let $\mathfrak{MS}$ denote the class of all Markov-static-split + static-split rules. For a rule $\left( \mathbf{P}, \bm{\alpha} = \left\lbrace \bm{\alpha}(j,h) \right\rbrace_{j \in \mathcal{J},h \in \mathcal{H}} \right) \in \mathfrak{MS}$, let $\bm{\sigma}$ denote the invariant probability distribution corresponding to the stochastic matrix $\mathbf{P}$. Then the expected switching and activation costs are given by $\mathsf{C}_0 \sum_{j', j \in \mathcal{J}} \sigma_{j'} P_{j',j} \norm{\left( j' - j \right)^+}_1$ and $\mathsf{C}_1 \sum_{j \in \mathcal{J}} \sigma_{j} \norm{j}_1$ respectively. We prove in the following theorem that the class $\mathfrak{MS}$ can achieve the same performance as $\mathfrak{M}$, the class of all ergodic Markov policies.
\begin{theorem}
\label{thm:markov-static-split}
For any $\bm{\lambda}$, $\bm{\mu}$ and $\varphi \in \mathfrak{M}$ such that $\bm{\lambda} \in \Lambda^{\varphi}(\bm{\mu})$, there exists a $\varphi' \in \mathfrak{MS}$ such that $\bm{\lambda} \in \Lambda^{\varphi'}(\bm{\mu})$ and $C^{\varphi'}(\bm{\mu}, \bm{\lambda}) = C^{\varphi}(\bm{\mu}, \bm{\lambda})$. Therefore,
\begin{align*}
C^{\mathfrak{M}}(\bm{\mu}, \bm{\lambda}) = \inf_{\varphi' \in \mathfrak{MS}, \lambda \in \Lambda^{\varphi'}(\bm{\mu})} C^{\varphi'}(\bm{\mu}, \bm{\lambda}).
\end{align*}
\end{theorem}
\begin{proof}[Proof Outline]
The proof of this theorem is similar to the proof of characterization of the stability region using the class of static-split policies. It maps the time-averages of BS activation transitions and rate allocations of the policy $\varphi \in \mathfrak{M}$ to a Markov-static-split rule $\varphi' \in \mathfrak{MS}$ that mimics the same time-averages. (Detailed proof is in the Appendix.)
\end{proof}
From the characterization of the class $\mathfrak{MS}$, Theorem~\ref{thm:markov-static-split} shows that the optimal cost $C^{\mathfrak{M}}(\bm{\mu}, \bm{\lambda})$ is equal to the optimal value of the optimization problem $\mathit{V}(\bm{\mu}, \bm{\lambda})$, which is given by
\begin{align*}
\inf_{\mathbf{P}, \bm{\alpha}} \mathsf{C}_0 \sum_{j', j \in \mathcal{J}} \sigma_{j'} P_{j',j} \norm{\left( j' - j \right)^+}_1 + \mathsf{C}_1 \sum_{j \in \mathcal{J}} \sigma_{j} \norm{j}_1
\end{align*}
such that $\mathbf{P} \in \mathcal{W}_{\card{\mathcal{J}}}$ with unique invariant distribution $\bm{\sigma}  \in \mathcal{P}_{\card{\mathcal{J}}}$,
 and $\bm{\alpha}(j,h) \in \mathcal{P}_{\card{\mathcal{R}(j,h)}} \, \forall j \in \mathcal{J},h \in \mathcal{H}$ with
\begin{align}
\label{eq:stability}
& \bm{\lambda} < \sum_{j \in \mathcal{J}} \sigma_j \sum_{h \in \mathcal{H}} \mu(h) \sum_{\mathbf{r} \in \mathcal{R}(j,h)} \alpha_{\mathbf{r}}(j,h) \mathbf{r}.
\end{align}


\subsection{A Modified Optimization Problem}
\label{subsec:linear-prog}
Now, consider the linear program given by
$$
\min_{\bm{\sigma}, \bm{\beta}} \mathsf{C}_1 \sum_{j \in \mathcal{J}} \sigma_{j} \norm{j}_1, \ \ \ \rm {such \ that}
$$
\begin{align*}
\bm{\sigma} & \in \mathcal{P}_{\card{\mathcal{J}}} \\
\beta_{j,h,\mathbf{r}} & \geq 0 \quad \forall \mathbf{r} \in \mathcal{R}(j,h),  \forall j \in \mathcal{J}, h \in \mathcal{H},	\\
\sigma_j & = \sum_{\mathbf{r} \in \mathcal{R}(j,h)} \beta_{j,h,\mathbf{r}}  \quad \forall j \in \mathcal{J}, h \in \mathcal{H}, \numberthis \label{eq:sigmaj}  \\
\bm{\lambda} & \leq \sum_{\substack{j \in \mathcal{J}, h \in \mathcal{H}, \\ \mathbf{r} \in \mathcal{R}(j,h)}}  \beta_{j,h,\mathbf{r}} \mu(h) \mathbf{r}. \numberthis \label{eq:pseudo-stability}
\end{align*}
The constraint (\ref{eq:sigmaj}) forces the right-hand side to be a constant over $h \in \mathcal{H}$.

Let $d := \card{\mathcal{J}} + \sum_{j \in \mathcal{J}, h \in \mathcal{H}} \card{\mathcal{R}(j,h)}$ be the number of variables in the above linear program. We denote by $\mathit{L}_{\mathbf{c}}(\bm{\mu}, \bm{\lambda})$, a linear program with constraints as above and with $\mathbf{c} \in \RR^d$ as the vector of weights in the objective function. Thus, the feasible set of the linear program $\mathit{L}_{\mathbf{c}}(\bm{\mu}, \bm{\lambda})$ is specified by the parameters $\bm{\mu}, \bm{\lambda}$ and the objective function is specified by the vector $\mathbf{c}$. Let $C^*_{\mathbf{c}}(\bm{\mu}, \bm{\lambda})$ denote the optimal value of $\mathit{L}_{\mathbf{c}}(\bm{\mu}, \bm{\lambda})$ and $\mathcal{O}^*_{\mathbf{c}}(\bm{\mu}, \bm{\lambda})$ denote the optimal solution set. Also, let
\begin{align*}
\mathcal{S} := \left\lbrace (\bm{\mu}, \bm{\lambda}) : \bm{\lambda} \in \Lambda(\bm{\mu}) \right\rbrace,
\end{align*}
\begin{align*}
\mathcal{U}_{\mathbf{c}} := \left\lbrace (\bm{\mu}, \bm{\lambda}) \in \mathcal{S} : \mathit{L}_{\mathbf{c}} (\bm{\mu}, \bm{\lambda}) \text{ has a unique solution} \right\rbrace.
\end{align*}
We claim that $\mathit{L}_{\mathbf{c}^0}(\bm{\mu}, \bm{\lambda})$, with
\begin{align}
\label{eq:actual-cost}
\mathbf{c}^0 := \left( (\mathsf{C}_1 \norm{j}_1)_{j \in \mathcal{J}}, \mathbf{0} \right)
\end{align}
provides a lower bound on the value of the original optimization problem $\mathit{V}(\bm{\mu}, \bm{\lambda})$. To see this, observe that we can lower bound the value by removing the switching cost from the objective. Then change variables $\beta_{j,h,{\bf r}} = \sigma_j \alpha_{\bf r}(j,h)$ to reach the new form, but with strict inequality in the last constraint on (\ref{eq:pseudo-stability}), and then relax this inequality.
Finally, (\ref{eq:sigmaj}) is met because $\sum_{{\bf r} \in \mathcal{R}(j,h)} \alpha_{\bf r}(j,h)=1$. These observations establish the claim. Therefore
\begin{align}
\label{eq:lp-opt}
C^*_{\mathbf{c}}(\bm{\mu}, \bm{\lambda}) \leq C^{\mathfrak{M}}(\bm{\mu}, \bm{\lambda}).
\end{align}
We use results from \cite{wets85lp-continuity}, \cite{davidson96stability-extreme} to show (in the Lemma below) that the solution set and the optimal value of the linear program are continuous functions of the input parameters.
\begin{lemma}
\label{lem:lp-continuity}
\begin{enumerate}[label=(\Roman*)]
\item \label{lem:lp-continuity-opt-value} As a function of the weight vector  $\mathbf{c}$ and the parameters $\bm{\mu}, \bm{\lambda}$, the optimal value $C^*_{(\cdot)}(\cdot)$ is continuous at any $(\mathbf{c}, (\bm{\mu}, \bm{\lambda})) \in \RR^d \times \mathcal{S}$.
\item \label{lem:lp-continuity-opt-set} For any weight vector $\mathbf{c}$, the optimal solution set $\mathcal{O}^*_{\mathbf{c}}(\cdot)$, as a function of the parameters $(\bm{\mu}, \bm{\lambda})$, is continuous at any $(\bm{\mu}, \bm{\lambda}) \in \mathcal{U}_{\mathbf{c}}$.
\end{enumerate}
\end{lemma}
\begin{remark}
Since $\mathcal{O}^*_{\mathbf{c}}(\bm{\mu}, \bm{\lambda})$ is a singleton if $(\bm{\mu}, \bm{\lambda}) \in \mathcal{U}_{\mathbf{c}}$, the definition of continuity in this context is unambiguous.
\end{remark}

\subsection{A Feasible Solution: Static-Split + Max-Weight}
\label{subsec:static-split+max-wt}
We now discuss how we can use the linear program $\mathit{L}$ to obtain a feasible solution for the original optimization problem (\ref{opt-problem}).
We need to deal with two modified constraints:

\noindent \textbf{(i) Single Ergodic Class -- Spectral Gap:}
For any $\bm{\sigma} \in \mathcal{P}_{\card{\mathcal{J}}}$ and $\epsilon_s \in (0,1)$, the stochastic matrix
\begin{align}
\label{eq:feasible-P-mx}
\mathbf{P}(\bm{\sigma}, \epsilon_s) := \epsilon_s \mathbf{1}_{\card{\mathcal{J}}} \bm{\sigma} + (1-\epsilon_s) \mathbf{I}_{\card{\mathcal{J}}}
\end{align}
 is aperiodic and has a single ergodic class given by $\{j: \sigma_j > 0 \}$ with $\bm{\sigma}$ as the invariant distribution. Therefore, given any optimal solution $\left( \bm{\sigma}, \bm{\beta} \right)$ for the relaxed problem $\mathit{L}_{\mathbf{c}}(\bm{\mu}, \bm{\lambda})$, we can construct a feasible solution $\left( \mathbf{P}(\bm{\sigma}, \epsilon_s), \bm{\alpha} \right)$ for the original optimization problem $\mathit{V}(\bm{\mu}, \bm{\lambda})$ such that the network cost for this solution is at most $\epsilon_s M \mathsf{C}_0$ more than the optimal cost. Note that $\epsilon_s$ is the spectral gap of the matrix $\mathbf{P}(\bm{\sigma}, \epsilon_s)$.


\noindent \textbf{(ii) Stability -- Capacity Gap:}
To ensure stability, it is necessary that the arrival rate is strictly less than the service rate (inequality~\eqref{eq:stability}). It can be shown that an optimal solution to the linear program satisfies the constraint~\eqref{eq:pseudo-stability} with equality, and therefore cannot guarantee stability. An easy remedy to this problem is to solve a modified linear program with a fixed small gap $\epsilon_g$ between the arrival rate and the offered service rate. We refer to the parameter $\epsilon_g$ as the \textit{capacity gap}. Continuity of the optimal cost of the linear program $\mathit{L}$ (from part~\ref{lem:lp-continuity-opt-value} of Lemma~\ref{lem:lp-continuity})  ensures that the optimal cost of the modified linear program is close to the optimal cost of the original optimization problem for sufficiently small $\epsilon_g$.

To summarize, if the statistical parameters $\bm{\mu}, \bm{\lambda}$ were known, one could adopt the following scheduling policy:\\ {\bf (a) BS activation}: Compute an optimal solution $\left( \bm{\sigma}^*, \bm{\beta}^* \right)$ for the linear program $\mathit{L}_{\mathbf{c}^0}(\bm{\mu}, \bm{\lambda}+\epsilon_g)$. At every time-slot, with probability $1-\epsilon_s$, maintain the BSs in the same state as the previous time-slot, i.e., no switching. With probability $\epsilon_s$, choose a new BS state according to the static-split rule given by $\bm{\sigma}^*$. The network can be operated at a cost close to the optimal by choosing $\epsilon_s$, $\epsilon_g$ sufficiently small.\\ {\bf (b) Rate allocation}: To ensure stability, use a queue-based rule such as the Max-Weight rule to allocate rates given the observed channel state:
    \begin{align}
    \label{eq:max-wt}
    \mathbf{S}(t) = \argmax_{\mathbf{r} \in \mathcal{R}\left( \mathbf{J}(t), H(t) \right)} \mathbf{Q}(t) \cdot \mathbf{r}.
    \end{align}
We denote the above static-split + Max-Weight rule with parameters $\epsilon_s$, $\epsilon_g$ by $\varphi(\bm{\mu}, \bm{\lambda}+\epsilon_g, \epsilon_s)$.
Theorem~\ref{thm:static-split-optimal} shows that the static-split + Max-Weight policy achieves close to optimal cost while ensuring queue stability.
\begin{theorem}
\label{thm:static-split-optimal}
For any $\bm{\mu}, \bm{\lambda}$ such that $(\bm{\mu}, \bm{\lambda}+2\epsilon_g) \in \mathcal{S}$, and for any $\epsilon_s \in (0,1)$, under the static-split + Max-Weight rule $\varphi(\bm{\mu}, \bm{\lambda}+\epsilon_g, \epsilon_s)$,
\begin{enumerate}
\item \label{item:cost-opt1} the network cost satisfies
\begin{align*}
C^{\varphi(\bm{\mu}, \bm{\lambda}+\epsilon_g, \epsilon_s)}(\bm{\mu}, \bm{\lambda}) \leq C^{\mathfrak{M}}(\bm{\mu}, \bm{\lambda}) + \kappa \epsilon_s + \gamma(\epsilon_g),
\end{align*}
for some constant $\kappa$ that depends on the network size and $\mathsf{C}_0$, $\mathsf{C}_1$, and for some increasing function $\gamma(\cdot)$ such that $\lim_{\epsilon_g \to 0} \gamma(\epsilon_g) = 0$, and
\item \label{item:stab1} the network is stable, i.e.,
 $$\bm{\lambda} \in \Lambda^{\varphi(\bm{\mu}, \bm{\lambda}+\epsilon_g, \epsilon_s)}(\bm{\mu}).$$
\end{enumerate}
\end{theorem}
\begin{proof}[Proof Outline]
Since $\mathbf{P}(\bm{\sigma}^*, \epsilon_s)$ has a single ergodic class, the marginal distribution of the activation state $(\mathbf{J}(t))_{t>0}$ converges to $\bm{\sigma}^*$. Part~\ref{item:cost-opt1} of the theorem then follows from  \eqref{eq:lp-opt} and the continuity of the optimal value of $\mathit{L}$ (Lemma~\ref{lem:lp-continuity}\ref{lem:lp-continuity-opt-value}). Part~\ref{item:stab1} relies on the strict inequality gap enforced by $\epsilon_g$ in \eqref{eq:stability}. Therefore, it is possible to serve all the arrivals in the long-term. We use a standard Lyapunov argument which shows that the $T$-step quadratic Lyapunov drift for the queues is strictly negative outside a finite set for some $T >0.$
A complete proof of this theorem can be found in the Appendix.
\end{proof}
One can also achieve the above guarantees with a static-split + static-split rule which has BS activations as above, but channel allocation through a static-split rule with convex combinations given by $\bm{\alpha}^*$ such that
\begin{align}
\label{eq:change-var}
\alpha^*_{\mathbf{r}}(j,h) = \frac{\beta^*_{j,h,\mathbf{r}} }{\sigma^*_j} \quad \forall \mathbf{r} \in \mathcal{R}(j,h),  \forall j \in \mathcal{J}, h \in \mathcal{H}.
\end{align}

\subsection{Effect of Parameter Choice on Performance}
\label{subsec:param}

The constants $\epsilon_s$ and $\epsilon_g$ can be used as control parameters to
 trade-off between two desirable but conflicting features --- small
 queue lengths and low network cost.


\noindent {\bf (i) Spectral gap, $\epsilon_s$}: $\epsilon_s$ is the spectral
  gap of the transition probability matrix  $\mathbf{P}(\bm{\sigma}^*,
  \epsilon_s)$ and, therefore, impacts the mixing time of the
  activation state $\left( \mathbf{J}(t) \right)_{t > 0}$. Since the
  average available service rate is dependent on the distribution of
  the activation state, the time taken for the queues to stabilize
  depends on the mixing time, and consequently, on the choice of
  $\epsilon_s$.  With  $\epsilon_s = 1,$ we are effectively ignoring
  switching costs, as this corresponds to a
  rule that chooses the activation sets in an i.i.d. manner according
  to the distribution $\bm{\sigma}^*$. Thus, stability is ensured but
  at a penalty of larger average costs. At the other extreme, when
  $\epsilon_s = 0$, the transition
  probability matrix $\mathbf{I}_{\card{\mathcal{J}}}$ corresponds to
  an activation rule that never switches the BSs from their initial
  activation state. This extreme naturally achieves zero switching cost,
  but does not guarantee queue stability as the initial activation set is frozen for all time and may not be large enough to ensure stable queues.

\noindent {\bf (ii) Capacity gap, $\epsilon_g$}: Recall that  $\epsilon_g$ is
  the gap enforced between the arrival rate and the allocated service
  rate in the linear program $\mathit{L}_{\mathbf{c}^0}(\bm{\mu},
  \bm{\lambda}+\epsilon_g)$. Since the mean queue-length is known to vary inversely as
  the capacity gap, the parameter  $\epsilon_g$ can be used to control
  queue-lengths. A small $\epsilon_g$ results in low network cost and
  large mean queue-lengths.


 \section{Policy with Unknown Statistics}
\label{sec:algo}
In the setting where arrival and channel statistics are unknown, our interest is in designing policies that learn the arrival and channel statistics to make rate allocation and BS activation decisions. As described in Section~\ref{subsec:allocation}, channel rates are observed in every time-slot after activation of the BSs. Since only channel rates of activated BSs can be obtained in any time-slot, the problem naturally involves a trade-off between activating more BSs to get better channel estimates versus maintaining low network cost. Our objective is to design policies that achieve network cost close to $C^{\mathfrak{M}}$, while learning the statistics well enough to stabilize the queues.

\subsection{An Explore-Exploit Policy}
\label{subsec:explore-exploit-policy}
\begin{algorithm}
  \caption{Policy $\phi(\epsilon_p, \epsilon_s, \epsilon_g)$ with parameters $\epsilon_p$, $\epsilon_s$,  $\epsilon_g$}
  \begin{algorithmic}[1]
  	\State Generate a uniformly distributed random direction $\bm{\upsilon} \in \RR^d$, $\norm{\bm{\upsilon}}_2 = 1$.
  	\State Construct a perturbed weight vector
  	$$\mathbf{c}^{\epsilon_p} \leftarrow \mathbf{c}^0 + \epsilon_p \bm{\upsilon}.$$
  	\State Initialize $\hat{\bm{\mu}} \leftarrow \mathbf{0}$, $\hat{\bm{\lambda}} \leftarrow \mathbf{0}$  and $\tilde{\mathbf{J}}(0) \leftarrow \mathbf{J}(0)$.
    \ForAll{$t > 0$}
    \State Generate $E_s(t)$, an indep. Bernoulli($\epsilon_s$) sample.
    \If { $E_s(t) = 0$}
    \Comment \textit{No Switching}
    \State $\tilde{\mathbf{J}}(t) \leftarrow \tilde{\mathbf{J}}(t-1)$.
    \Else
    \State  Solve $\mathit{L}_{\mathbf{c}^{\epsilon_p}} \bigl( \hat{\bm{\mu}}, \hat{\bm{\lambda}}+\epsilon_g \bigl)$.
    \State Select an optimal solution $\bigl( \hat{\bm{\sigma}}(t), \hat{\bm{\beta}}(t) \bigr)$.
    \State Select $\tilde{\mathbf{J}}(t)$ according to the distribution $\hat{\bm{\sigma}}(t)$.
    \EndIf
    \State Set $\epsilon_l(t)  \leftarrow \frac{2\log t}{t}.$
    \State Generate $E_l(t)$, an indep. Bernoulli($\epsilon_l(t)$) sample.
    \If { $E_l(t) = 1$}
    \Comment \textit{Explore}
    \State $\mathbf{J}(t) \leftarrow \mathbf{1}$ (Activate all the BSs).
    \State Observe the channel state $H(t)$.
    \State Update empirical distributions $\hat{\bm{\mu}}$, $\hat{\bm{\lambda}}$.
    \Else
    \Comment \textit{Exploit}
    \State $\mathbf{J}(t) \leftarrow \tilde{\mathbf{J}}(t)$.
    \State Observe the channel state $H(t)|_{\mathbf{J}(t)}$.
    \EndIf
    \State Allocate channels according to the Max-Weight Rule,
    \begin{align*}
    \mathbf{S}(t)  \leftarrow \argmax_{\mathbf{r} \in \mathcal{R}\left( \mathbf{J}(t), H(t) \right)} \mathbf{Q}(t) \cdot \mathbf{r}.
    \end{align*}
 \EndFor
  \end{algorithmic}
    \label{alg:}
\end{algorithm}
Algorithm~\ref{alg:} gives a policy $\phi(\epsilon_p, \epsilon_s, \epsilon_g)$, which is an explore-exploit strategy similar to the $\epsilon$-greedy policy in the multi-armed bandit problem. Here, $\epsilon_p, \epsilon_s, \epsilon_g$ are fixed parameters of the policy. If an iterative scheme is used to solve the LP (line 15 of Algorithm \ref{alg:}), one could initialize the iteration at the solution parameterized by the previously obtained empirical distributions.
\subsubsection{Initial Perturbation of the Cost Vector}
Given the original cost vector $\mathbf{c}^0$ (given by \eqref{eq:actual-cost}), the policy first generates a slightly perturbed cost vector $\mathbf{c}^{\epsilon_p}$ by adding to $\mathbf{c}^0$ a random perturbation uniformly distributed on the $\epsilon_p$-ball. It is easily verified that, for any $(\bm{\mu}, \bm{\lambda}) \in \mathcal{S}$,
\begin{align*}
\card{C^*_{\mathbf{c}^{\epsilon_p}}(\bm{\mu}, \bm{\lambda}) - C^*_{\mathbf{c}^0}(\bm{\mu}, \bm{\lambda})} \leq \sqrt{\card{\mathcal{H}}+1} \mathsf{C}_1 \epsilon_p.
\end{align*}
In addition, the following lemma shows that the perturbed linear program has a unique solution with probability 1.
\begin{lemma}
\label{lem:perturbed-unique-soln}
For any $(\bm{\mu}, \bm{\lambda}) \in \mathcal{S}$,
\begin{align*}
\PP\left[(\bm{\mu}, \bm{\lambda}) \in \mathcal{U}_{\mathbf{c}^{\epsilon_p}}  \given \mathbf{J}(0), \mathbf{Q}(1) \right] = 1.
\end{align*}
\end{lemma}
\subsubsection{BS Activation}
\paragraph*{\bf Estimated Markov-static-split rule}
The policy attempts to mimic the Markov-static-split rule using the empirical means  $(\hat{\bm{\mu}}, \hat{\bm{\lambda}})$. The vector $\tilde{\mathbf{J}}(t)$ is used to keep track of the BS activations according to the estimated Markov-static-split rule. To be precise, with probability $1-\epsilon_s$, the policy chooses to keep the same activation set as the previous time-slot's candidate, i.e., $\tilde{\mathbf{J}}(t-1)$. With probability $\epsilon_s$, it solves the linear program $\mathit{L}_{\mathbf{c}^{\epsilon_p}}\bigl( \hat{\bm{\mu}}, \hat{\bm{\lambda}}+\epsilon_g \bigr)$ with the perturbed cost vector $\mathbf{c}^{\epsilon_p}$ and parameters $\hat{\bm{\mu}}, \hat{\bm{\lambda}}+\epsilon_g$ given by the empirical distribution. From an optimal solution $\bigl( \hat{\bm{\sigma}}(t), \hat{\bm{\beta}}(t) \bigr)$ of the linear program, it chooses the BS activation vector $\tilde{\mathbf{J}}(t)$ according to the distribution $\hat{\bm{\sigma}}(t)$.

\paragraph*{\bf Explore-Exploit}
At each time, the policy chooses to either \textit{explore} or \textit{exploit} and accordingly selects the actual BS activation vector $\mathbf{J}(t)$. The probability that it explores, $\epsilon_l(t) = \frac{2\log t}{t}$, decreases with time.
\begin{itemize}
\item In the \textit{explore} phase, the policy activates all the BSs and observes the channel. It maintains $\hat{\bm{\mu}}, \hat{\bm{\lambda}}$, the empirical distribution of the channel and the empirical mean of the arrival vector respectively, obtained from samples in the explore phase.
\item In the \textit{exploit} phase, it simply chooses the activation vector given by the estimated Markov-static-split rule, i.e., $\mathbf{J}(t) = \tilde{\mathbf{J}}(t)$. 
\end{itemize}

\subsubsection{Rate Allocation}
The policy uses the Max-Weight Rule given by \eqref{eq:max-wt} for channel allocation.

\subsection{Performance Guarantees}
In Theorem~\ref{thm:algo-optimal}, we give stability and network cost guarantees for the proposed learning-cum-scheduling rule $\phi(\epsilon_p, \epsilon_s, \epsilon_g)$.
\begin{theorem}
\label{thm:algo-optimal}
For any $\bm{\mu}, \bm{\lambda}$ such that $(\bm{\mu}, \bm{\lambda}+2\epsilon_g) \in \mathcal{S}$, and for any $\epsilon_p, \epsilon_s \in (0,1)$, under the policy $\phi(\epsilon_p, \epsilon_s, \epsilon_g)$,
\begin{enumerate}
\item \label{item:cost-opt2} the network cost satisfies
\begin{align*}
C^{\phi(\epsilon_p, \epsilon_s, \epsilon_g)}(\bm{\mu}, \bm{\lambda}) \leq C^{\mathfrak{M}}(\bm{\mu}, \bm{\lambda}) + \kappa (\epsilon_p + \epsilon_s) + \gamma(\epsilon_g),
\end{align*}
for some constant $\kappa$ that depends on the network size and $\mathsf{C}_0$, $\mathsf{C}_1$, and for some increasing function $\gamma(\cdot)$ such that $\lim_{\epsilon_g \to 0} \gamma(\epsilon_g) = 0$, and
\item \label{item:stab2} the network is stable, i.e.,
 $$\bm{\lambda} \in \Lambda^{\phi(\epsilon_p, \epsilon_s, \epsilon_g)}(\bm{\mu}).$$
\end{enumerate}
\end{theorem}
\begin{proof}[Proof Outline]
  As opposed to known statistical parameters for the arrivals and the
  channel in the Markov-static-split rule, the policy uses empirical
  statistics that change dynamically with time. Thus, the activation
  state process $(\mathbf{J}(t))_{t>0}$, in this case, is not a
  time-homogeneous Markov chain. However, we note that $\mathbf{J}(t)$
  along with the empirical statistics forms a time-inhomogeneous
  Markov chain with the empirical statistics converging to the
  true statistics almost surely. Specifically, we show that the time
  taken by the algorithm to learn the parameters within a small
  error has a finite second moment.

  We then use convergence results for time-inhomogeneous Markov chains
  (derived in Lemma~\ref{lem:inhomo-dtmc-convergence} in
  Section~\ref{sec:inhom-markov}) to show convergence of the marginal
  distribution of the activation state $(\mathbf{J}(t))_{t>0}$. As in
  Theorem~\ref{thm:static-split-optimal}, Part~\ref{item:cost-opt2}
  then follows from \eqref{eq:lp-opt} and the continuity of the
  optimal value of $\mathit{L}$ (Lemma~\ref{lem:lp-continuity}\ref{lem:lp-continuity-opt-value}).

  Part~\ref{item:stab2} requires further arguments. The queues have a
  negative Lyapunov drift only after the empirical estimates have
  converged to the true parameters within a small error. To bound the
  Lyapunov drift before this time, we use boundedness of the arrivals
  along with the existence of a second moment for the convergence time
  of the estimated parameters. By using a telescoping argument as in
  Foster's theorem, we show that this implies stability as per
  Definition~\ref{def:stability}. For the complete proof, please see the Appendix.
\end{proof}

\subsection{Optimality of static-split + max-weight policies}
\label{subsec:discussion-optimality}
We now address the restriction to ergodic Markov policies and the question of its optimality. Recall Definition~\ref{def:stability}. When there is only activation cost, and no switching cost, it is easy to see that the class of static-split policies is both cost and throughput optimal. This scenario is represented by the linear program $\mathit{L}_{\mathbf{c}}(\bm{\mu}, \bm{\lambda})$ in Section~\ref{subsec:linear-prog}. The optimal cost for the problem with switching cost cannot be lower than the value of $\mathit{L}_{\mathbf{c}}$; see~(\ref{eq:lp-opt}). The static-split + max-weight policies in Section~\ref{subsec:static-split+max-wt} can get arbitrarily close to this value; see Theorem~\ref{thm:static-split-optimal}. We can thus conclude that it is sufficient to consider the class of all ergodic Markov policies. Theorem~\ref{thm:algo-optimal} finally asserts that a Markov-static-split policy for BS activation can be implemented using estimated parameters.

\subsection{Discussion: Other Potential Approaches}
\label{optFramework-Discussion-VirtualQueues}

Recall that our system consists of two distinct time-scales: (a)
exogenous fast dynamics due the channel variability, that occurs on a
per-time-slot basis, and (b) endogenous slow dynamics of learning and
activation due to base-station active-sleep state dynamics. By
`exogenous', we mean that the time-scale is controlled by nature
(channel process), and by `endogenous', we mean that the time-scale is
controlled by the learning-cum-activation algorithm (slowed dynamics
where activation states change only infrequently).  To place this in
perspective, consider the following alternate approaches, each of
which has defects.

{\em 1. Virtual queues + MaxWeight:} As is now standard
\cite{GeoNeeTas_06,sryi14}, suppose that we encode the various costs
through virtual queues (or variants there-of), and apply a MaxWeight
algorithm to this collection of queues. Due to the switching cost, the
effective channel, i.e., the vector of channel rates on the active
collection of base-stations, has dependence across time (coupled
dynamics of channel and queues) through the activation set scheduling,
and voids the standard Lyapunov proof approach for showing
stability. Specifically, we cannot guarantee that the time average of
various activation sets chosen by this (virtual + actual queue)
MaxWeight algorithm equals the corresponding optimal fractions computed
using a linear program with known channel and arrival parameters.

{\em 2. Ignoring Switching Costs with Fast Dynamics:} Suppose we use
virtual queues to capture only the activation costs. In this case, a
MaxWeight approach (selecting a new activation set and channel
allocation in each time-slot) will ensure stability, but will not
provide any guarantees on cost optimality as there will be frequent
switching of the activation set.

{\em 3. Ignoring Switching Costs with Slowed Dynamics:} Again, we use
virtual queues for encoding only activation costs, and use block
scheduling. In other words, re-compute an activation + channel
schedule once every $R$ time-slots, and use this fixed schedule for
this block of time (pick-and-compare, periodic, frame-based algorithms
\cite{Tas_98,Neely02tradeoffsin,chsa06,yiprch08}). While this approach
minimizes switching costs (as activation changes occur infrequently),
stability properties are lost as we are not making use of opportunism
arising from the wireless channel variability (the schedule is fixed
for a block of time and does not adapt to instantaneous channel
variations).

Our approach avoids the difficulties in each of these approaches by
explicitly slowing down the time-scale of the activation set dynamics
(an engineered slow time-scale), thus minimizing switching
costs. However, it allows channels to be opportunistically
re-allocated in each time-slot based on the instantaneous channel state
(the fast time-scale of nature). The channel allocations are based on observations of channel state but only on the activated BSs. This fast-slow co-evolution of
learning, activation sets and queue lengths requires a new proof
approach. We combine new results (see Section~\ref{sec:inhom-markov})
on convergence of inhomogeneous Markov chains with Lyapunov analysis
to show both stability and cost (near) optimality.

 \section{Convergence of a Time-Inhomogeneous Markov Process}
\label{sec:inhom-markov}
 In this section, we derive some convergence bounds for perturbed time-inhomogeneous Markov chains which are useful in proving stability and cost optimality. Let $\mathscr{P} := \{\mathbf{P}_{\delta}, \delta\in\Delta\}$ be a collection of stochastic matrices in $\RR^{N\times N}$, with $\{\bm{\sigma}_\delta, \delta\in\Delta\}$ denoting the corresponding invariant probability distributions. Also, let $\mathbf{P}_*$ be an $N\times N$ aperiodic stochastic matrix with  a single ergodic class and invariant probability distribution $\bm{\sigma}_*$. 

Recall that for a stochastic matrix $\mathbf{P}$ the coefficient of ergodicity \cite{seneta2006non}
$\tau_1(\mathbf{P})$ is defined by
\begin{align}\label{E-tau1}
\tau_1(\mathbf{P}) \;:=\; \max_{\mathbf{z}^{\mathsf{T}}\mathbf{1}_N=0\,,\;\norm{\mathbf{z}}_1=1}\; \norm{\mathbf{P}^{\mathsf{T}} \mathbf{z}}_1.
\end{align}
It has the following basic properties \cite{seneta2006non}: 
\begin{enumerate}
\item $\tau_1(\mathbf{P}_1\mathbf{P}_2) \;\le\; \tau_1(\mathbf{P}_1) \tau_1(\mathbf{P}_2)$,
\item $\card{\tau_1(\mathbf{P}_1)-\tau_1(\mathbf{P}_2)} \;\le\; \norm{\mathbf{P}_1-\mathbf{P}_2}_\infty$,
\item $\norm{\mathbf{x} \mathbf{P} -  \mathbf{y}\mathbf{P}}_1\;\le\; \tau_1(\mathbf{P})\,\norm{\mathbf{x}-\mathbf{y}}_1  	\quad
\forall\, \mathbf{x},\mathbf{y}\in \mathcal{P}_N$,
 and
\item $\tau_1(\mathbf{P})<1$ if and only if $\mathbf{P}$ has no pair of orthogonal rows
(i.e., if it is a scrambling matrix).

\end{enumerate}

By  the results in \cite{anthonisse1977exponential}, if $\mathbf{P}_*$ is aperiodic
and has a single ergodic class then there exists an integer $\Hat{m}$ such that
$\mathbf{P}_*^k$ is scrambling for all $k\ge\Hat{m}$.
Therefore, $\tau_1(\mathbf{P}_*^{k}) < 1 \; \forall k \geq \Hat{m}$.

Define
\begin{align}
\label{eq:def-sup-norm}
\epsilon :=  \sup_{\delta\in\Delta} \, \norm{\mathbf{P}_{\delta}-\mathbf{P}_*}_1.
\end{align}

Now, consider a time-inhomogeneous Markov chain $\left( X(t)
\right)_{t \geq 0}$ with initial distribution  $\mathbf{y}(0)$, and
transition probability matrix at time $t$ given by
$\mathbf{P}_{\delta_t} \in \mathscr{P} \; \forall t > 0$. Let
$\left\lbrace \mathbf{y}(t) \right\rbrace_{t \geq 0}$ be the resulting
sequence of marginal distributions. The following lemma gives a bound
on the convergence of the limiting distribution of such a
time-inhomogeneous DTMC to $\bm{\sigma}_*$. Additional results are
available in the Appendix.
\begin{lemma}
\label{lem:inhomo-dtmc-convergence}
For any $\mathbf{y}(0)$, 
\begin{enumerate}[label=(\alph*)]
\item \label{item:marginal-bound} the marginal distribution satisfies
\begin{align}
\label{eq:marginal-bound}
\norm{\mathbf{y}(n)-\bm{\sigma}_*}_1 &\;\le\; \tau_1(\mathbf{P}_*^n) \norm{\mathbf{y}(0)-\bm{\sigma}_*}_1\nonumber\\
&\mspace{100mu}+ \epsilon\, \sum_{\ell=0}^{n-1} \tau_1(\mathbf{P}_*^\ell)\, ,
\end{align}
\item \label{item:limiting-bound} and the limiting distribution satisfies
\begin{align*}
\limsup_{n \to \infty} \, \norm{\mathbf{y}(n) - \bm{\sigma}_*}_1 \;\le\;  \epsilon \, \Upsilon(\mathbf{P}_*)
\end{align*}
where $\Upsilon(\mathbf{P}_*)  := \sum_{\ell=0}^{\infty} \tau_1(\mathbf{P}_*^\ell)\, \leq \frac{\Hat{m}}{1-\tau_1(\mathbf{P}_*^{\Hat{m}})}.$
\end{enumerate}

\end{lemma}

\begin{proof}
The trajectory $(\mathbf{y}(n))_{n>0}$ satisfies $\forall n \geq 1,$
\begin{align}\label{E-traj}
\mathbf{y}(n)\;=\; \mathbf{y}(n-1) \mathbf{P}_* + \mathbf{y}(n-1)(\mathbf{P}_{\delta_{n-1}}-\mathbf{P}_*).
\end{align}
%
%
%
Using \eqref{E-traj} recursively, we have
\begin{align*}
\mathbf{y}(n) 
&\;=\; \mathbf{y}(0) \mathbf{P}_*^n +\sum_{k=1}^n \mathbf{y}(n-k) (\mathbf{P}_{\delta_{n-k}}-\mathbf{P}_*)\mathbf{P}_*^{k-1},
\end{align*}
which gives us
\begin{align}
\label{eq:deviation}
\mathbf{y}(n)-\bm{\sigma}_* & = (\mathbf{y}(0)-\bm{\sigma}_*) \mathbf{P}_*^n	\n \\
& \quad +\sum_{k=1}^n \mathbf{y}(n-k) (\mathbf{P}_{\delta_{n-k}}-\mathbf{P}_*)\mathbf{P}_*^{k-1}.
\end{align}
 Now, taking norms and using the definitions in \eqref{E-tau1}
and \eqref{eq:def-sup-norm}, we obtain 
\begin{align*}
\norm{\mathbf{y}(n)-\bm{\sigma}_*}_1\;\le\; \tau_1(\mathbf{P}_*^n) \norm{\mathbf{y}(0)-\bm{\sigma}_*}_1
+ \epsilon\, \sum_{\ell=0}^{n-1} \tau_1(\mathbf{P}_*^\ell)\,.
\end{align*}
This proves part~\ref{item:marginal-bound} of the lemma. Now, note that
\begin{align}\label{E4.7}
\tau_1(\mathbf{P}_*^k)\;\le\;
\bigl(\tau_1(\mathbf{P}_*^{m})\bigr)^{\lfloor \nicefrac{k}{m}\rfloor}
\end{align}
for any positive integers $k$, $m$.
Since $\tau_1(\mathbf{P}_*^{\Hat{m}})<1$, it follows that  $\lim_{n \to \infty} \tau_1(\mathbf{P}_*^n) = 0$, and
$$\Upsilon(\mathbf{P}_*) \,=\, \sum_{\ell=0}^{\infty} \tau_1(\mathbf{P}_*^\ell)\, \le \frac{\Hat{m}}{1-\tau_1(\mathbf{P}_*^{\Hat{m}})}.$$
Using this in \eqref{eq:marginal-bound}, we have
\begin{align*}
\limsup_{n \to \infty} \norm{\mathbf{y}(n)-\bm{\sigma}_*}_1 
\leq \epsilon\,\Upsilon(\mathbf{P}_*) \leq \frac{\epsilon\,\Hat{m}}{1-\tau_1(\mathbf{P}_*^{\Hat{m}})}\,,
\end{align*}
which proves part~\ref{item:limiting-bound} of the lemma.
\end{proof}

 \section{Simulation Results}
\label{sec:sims}

\begin{figure}[!b]
\begin{center}
\includegraphics[width=3.5in]{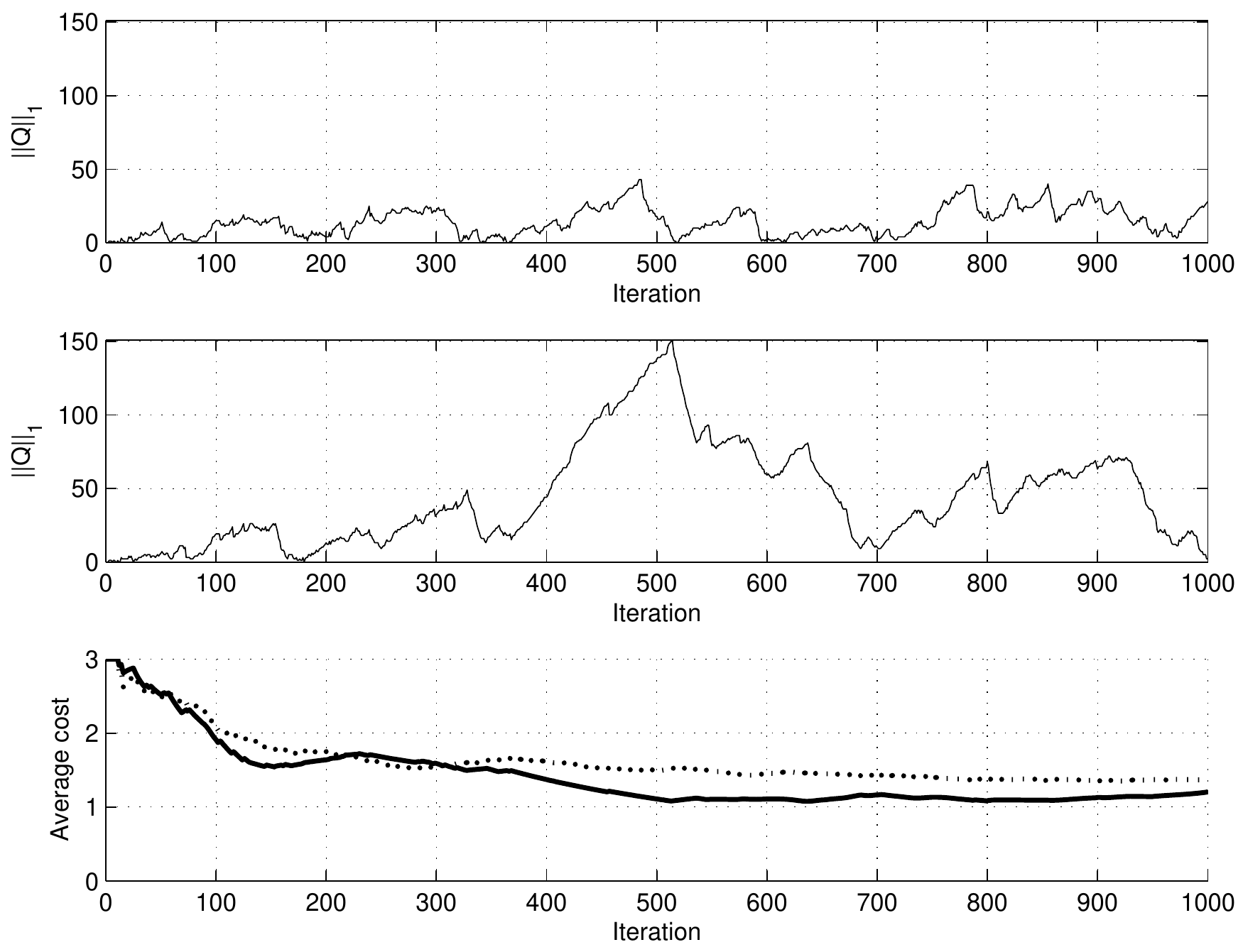}
\caption{The top two plots show the total queue size as a function of time when $\epsilon_s = 0.2$ and $\epsilon_s = 0.05$, respectively. The bottom plot shows the corresponding average costs (with the solid curve for $\epsilon_s = 0.05$). A smaller $\epsilon_s$ yields a lower average cost but has higher queue occupancy.}
\label{fig:simulation}
\end{center}
\end{figure}

\begin{figure}[!b]
\begin{center}
\includegraphics[width=3.5in]{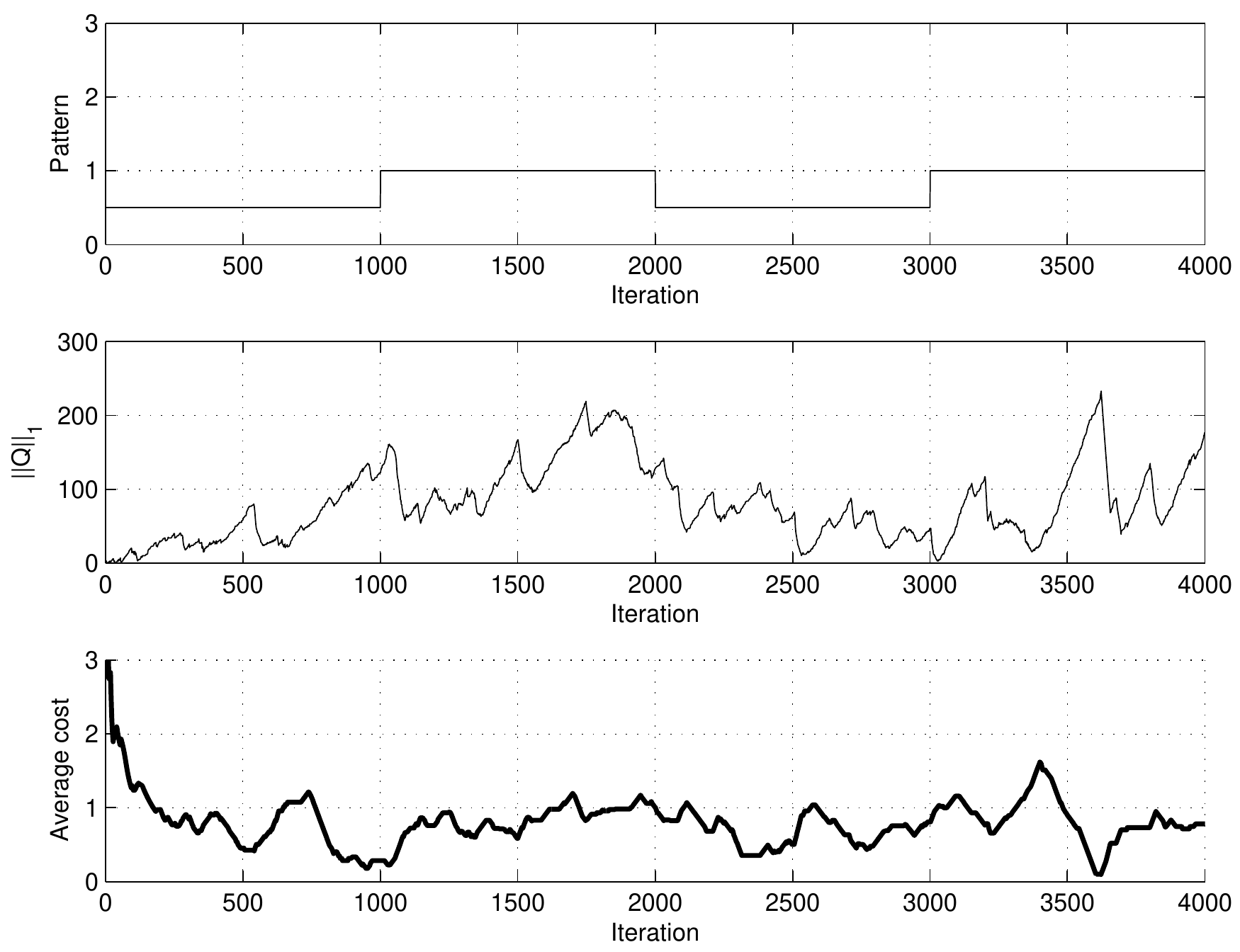}
\caption{The top plot shows a time-varying traffic pattern. The middle plot shows total queue size as a function of time when $\epsilon_s = 0.05$. The bottom plot shows the corresponding short-term averaged cost. The learning algorithm is modified and employs a constant learning rate so as to track the regime changes. Due to a constant learning rate, queue occupancy is a little higher, but the algorithm tracks the changes and stabilizes queue. The short-term average cost is kept small through the regime changes. The larger fluctuation in comparison to the bottommost plot in Figure \ref{fig:simulation} is due to the short-term nature of the average.}
\label{fig:simulation2}
\end{center}
\end{figure}

We present simulations that corroborate the theoretical results in this paper. The setting is as follows. There are five users and three BSs in the system. BS 1 can service users 1, 2, and 5. BS 2 can service users 1, 2, 3, and 4. BS 3 can service users 3, 4, and 5. The Bernoulli arrival rates on each queue (which have to be learned by the algorithm) is 0.1 packets/slot on each mobile-BS service connection. The total arrival rate to the system is thus 0.1 packet/slot $\times$ 10 connections, or 1 packet/slot. A good channel yields a service of 2 packets/slot while a bad channel yields 1 packet/slot. In our correlated fading model, either all channels are bad, or all connections to exactly one BS are good while the others bad. This yields four correlated channel states and all four are equiprobable (the probabilities being unknown to the algorithm). The fading process is independent and identically distributed over time. The activation constraint is that each BS can service at most one mobile per slot. The per BS  switching cost $ \mathsf{C}_0$ and activation cost $ \mathsf{C}_1$ are both taken to be 1.

Figure \ref{fig:simulation} provides the instantaneous queue sizes (first two plots) and time-averaged costs (third plot) for two values of $\epsilon_s$, namely, $0.2$ (first plot) and $0.05$ (second plot). The plots show that a smaller $\epsilon_s$ yields a lower average cost and stabilizes the queue, but has higher queue occupancy.

Figure \ref{fig:simulation2} considers a situation with regime changes (see top plot). A value 0.5 indicates that all instantaneous arrival rates are lowered by a factor 0.5. The parameter $\epsilon_s = 0.05$. The middle plot shows instantaneous total queue occupancy. The bottom plot is a short-term average cost (averaged over the past 200 slots). The algorithm was modified to keep the learning rate for estimating $\hat{\bf \lambda}$ and $\hat{\bf \mu}$ not below a threshold $(0.001)$ to help track regime changes. 
Figure \ref{fig:simulation2} indicates that the queues are stabilized but have a higher occupancy due to the use of a constant learning rate in comparison to the middle plot in Figure \ref{fig:simulation}. But the short-term average cost (bottom plot) is kept small through the regime changes.

\section{Conclusion}
\label{sec:conclusion}
We study the problem of jointly activating base-stations along with
channel allocation, with the objective of minimizing energy costs
(activation + switching) subject to packet queue stability. Our
approach is based on timescale decomposition, consisting of fast-slow
co-evolution of user queues (fast) and base-station activation sets
(slow). We develop a learning-cum-scheduling algorithm that can
achieve an average cost that is arbitrarily close to optimal, and
simultaneously stabilize the user queues (shown using convergence results for inhomogeneous Markov chains).

\section*{Acknowledgements}
This work was partially supported by NSF grants CNS-1017549,
CNS-1161868, CNS-1343383, CNS-1731658 and DMS-1715210, Army Research Office grant W911NF-17-1-0019,
the US DoT supported D-STOP Tier 1 University Transportation Center,
and the Robert Bosch Centre for Cyber Physical Systems.

\appendices

\section{Proof of Theorem~\ref{thm:markov-static-split} }
\begin{proof}[Proof of Theorem~\ref{thm:markov-static-split}]
Consider an ergodic Markov policy $\varphi \in \mathfrak{M}$ such that $\bm{\lambda} \in \Lambda^{\varphi}(\bm{\mu})$. We use the notation $\PP_{\pi}$ to denote probabilities corresponding to the stationary distribution under policy $\varphi$. Let for all $j',j \in \mathcal{J}$,
\begin{align*}
\sigma_j & = \PP_{\pi}\left[ \mathbf{J}(t) = j \right],	\\
P_{j',j} & =  \PP_{\pi}\left[ \mathbf{J}(t) = j \given \mathbf{J}(t-1) = j' \right] \ind{\sigma_j \sigma_{j'} > 0},
\end{align*}
and
\begin{align*}
\alpha_{\mathbf{r}}(j, h) & = \PP_{\pi}\left[ \mathbf{S}(t) = \mathbf{r} \given \mathbf{J}(t) = j, H(t) = h  \right] \\
& \quad \forall \mathbf{r} \in \mathcal{R}(j,h), \; \forall j \in \mathcal{J},  h \in \mathcal{H}.
\end{align*}
We first show that $\mathbf{P} \in \mathcal{W}_{\card{\mathcal{J}}}$. Since $\varphi \in \mathfrak{M}$, the network state process $\{X(t)\}_{t > 0}$, where $X(t) = \left( \mathbf{J}(t-1), \mathbf{Q}(t) \right)$, under policy $\varphi$ is an ergodic Markov chain. Therefore, for any $j',j \in \mathcal{J}$ such that $\sigma_j, \sigma_{j'} > 0$, there exists a constant $k \in \NN$ and $k$ states $(j_0 = j', \mathbf{q}_1), (j_1, \mathbf{q}_2), \dots, (j_{k-1} = j, \mathbf{q}_k) \in \mathcal{J} \times \mathbb{Z}^{M \times n}$ such that for all $1 \leq l \leq k-1$,
\begin{align*}
\PP_{\pi}\left[ X(l+1) = (j_l, \mathbf{q}_{l+1}) \given X(l) = (j_{l-1}, \mathbf{q}_{l}) \right] > 0,
\end{align*}
and
\begin{align*}
\PP_{\pi}\left[ X(l) = (j_{l-1}, \mathbf{q}_{l}) \right] > 0.
\end{align*}
This gives us that
\begin{align*}
P_{j_{l-1},j_l} & =  \PP_{\pi}\left[ \mathbf{J}(l) = j_l \given \mathbf{J}(l-1) = j_{l-1} \right] > 0.
\end{align*}
Therefore, for any $j',j \in \mathcal{J}$ such that $j \neq j'$ and $\sigma_j, \sigma_{j'} > 0$, there exists a $k \in \NN$ such that $P^{k}_{j',j} > 0$. A similar argument shows that $P$ is aperiodic. In addition, we also have $\bm{\sigma} = \bm{\sigma} \mathbf{P}$, from which we can conclude that $\mathbf{P} \mathbf{1}_{\mathcal{J}} = \mathbf{1}_{\mathcal{J}}$. This proves that  $\mathbf{P}$ is a stochastic matrix with a single ergodic class, i.e., $\mathbf{P} \in \mathcal{W}_{\card{\mathcal{J}}}$.

Further, it is easy to verify that
\begin{align*}
\mathsf{C}_0 \sum_{j', j \in \mathcal{J}} \sigma_{j'} P_{j',j} \norm{\left( j' - j \right)^+}_1 + \mathsf{C}_1 \sum_{j \in \mathcal{J}} \sigma_{j} \norm{j}_1 =  C^{\varphi}(\bm{\mu}, \bm{\lambda}),
\end{align*}
and
\begin{align*}
\sum_{j \in \mathcal{J}} \sigma_j \sum_{h \in \mathcal{H}} \mu(h) \sum_{\mathbf{r} \in \mathcal{R}(j,h)} \alpha_{\mathbf{r}}(j,h) \mathbf{r} = \EE_{\pi}\left[ \mathbf{S}(t) \right].
\end{align*}
Since $\bm{\lambda} \in \Lambda^{\varphi}(\bm{\mu})$, we have $\bm{\lambda} < \EE_{\pi}\left[ \mathbf{S}(t) \right]$.
Therefore, for $$\varphi' := \left( \mathbf{P}, \bm{\alpha} = \left\lbrace \bm{\alpha}(j,h) \right\rbrace_{j \in \mathcal{J},h \in \mathcal{H}} \right),$$ we have $\varphi' \in \mathfrak{MS}$, $\bm{\lambda} \in \Lambda^{\varphi'}(\bm{\mu})$ and $C^{\varphi'}(\bm{\mu}, \bm{\lambda}) = C^{\varphi}(\bm{\mu}, \bm{\lambda})$.
\end{proof}

\section{Proofs of  Lemmas~\ref{lem:lp-continuity} and \ref{lem:perturbed-unique-soln}}
\begin{proof}[Proof of Lemma~\ref{lem:lp-continuity}]
Let $\mathcal{F}$ and $\mathcal{D}$ denote the feasible sets of $\mathit{L}$ and its dual respectively. By Theorem 2 in \cite{wets85lp-continuity}, to prove \ref{lem:lp-continuity-opt-value}, it is sufficient to establish that $\mathcal{F}$ and $\mathcal{D}$ are continuous multifunctions on $\RR^d \times \mathcal{S}$. The feasible set of the linear program depends only on $(\bm{\mu}, \bm{\lambda})$ and not on $\mathbf{c}$. By Proposition 6 in \cite{wets85lp-continuity}, $\mathcal{F}$ is continuous on $\mathcal{S}$ if
\begin{enumerate}[label=(\roman*)]
\item the dimension of $\mathcal{F}$ is constant on $\mathcal{S}$, and
\item for any $\left( \bm{\mu}, \bm{\lambda} \right) \in \mathcal{S}$, there exists a neighborhood $\mathcal{V}$ of $\left( \bm{\mu}, \bm{\lambda} \right)$ such that, if a particular inequality constraint 
  is tight (satisfied with equality) for all $x \in \mathcal{F}\left( \bm{\mu}, \bm{\lambda} \right)$, then for any $\left( \bm{\mu}', \bm{\lambda}' \right) \in \mathcal{V}$, the corresponding constraint
  is tight for all $x \in \mathcal{F}\left( \bm{\mu}', \bm{\lambda}' \right)$.
\end{enumerate}
The above two conditions are satisfied if
\begin{enumerate}[label=(\roman*)]
\item the equality constraints are the same for every $\mathcal{F}\left( \bm{\mu}, \bm{\lambda} \right)$, and
\item for any $\left( \bm{\mu}, \bm{\lambda} \right) \in \mathcal{S}$, no inequality constraint is tight for every $x \in \mathcal{F}\left( \bm{\mu}, \bm{\lambda} \right)$. 
\end{enumerate}
These can be verified to be true for all $\left( \bm{\mu}, \bm{\lambda} \right) \in \mathcal{S}$. Therefore,  $\mathcal{F}$ is continuous on $\mathcal{S}$.

According to Corollary 11 in \cite{wets85lp-continuity}, $\mathcal{D}$ is continuous on $\RR^d \times \mathcal{S}$ if $\mathcal{F}$ is bounded
. This is again true since any feasible solution is a set of probability mass functions. Therefore, by Theorem 2 in \cite{wets85lp-continuity}, $C^*$ is continuous on $\RR^d \times \mathcal{S}$.

To prove \ref{lem:lp-continuity-opt-set}, i.e., that the optimal solution set $\mathcal{O}^*_{\mathbf{c}}(\cdot)$ is continuous on $\mathcal{U}_{\mathbf{c}}$, we first note that $\mathcal{O}^*_{\mathbf{c}}\left( \bm{\mu}, \bm{\lambda} \right)$ is the feasible set for a family of linear constraints, which is same as that for $\mathcal{F}\left( \bm{\mu}, \bm{\lambda} \right)$, in addition to the equality constraint
\begin{align*}
\mathbf{c} \cdot\left( \bm{\sigma}, \bm{\beta} \right)  = C^*_{\mathbf{c}}\left( \bm{\mu}, \bm{\lambda} \right).
\end{align*}
By definition, the set $\mathcal{O}^*_{\mathbf{c}}\left( \bm{\mu}, \bm{\lambda} \right)$ is non-empty for any $\left( \bm{\mu}, \bm{\lambda} \right) \in \mathcal{S}$, and is a singleton for any $\left( \bm{\mu}, \bm{\lambda} \right) \in \mathcal{U}_{\mathbf{c}}$. Now consider any $\left( \bm{\mu}, \bm{\lambda} \right) \in \mathcal{U}_{\mathbf{c}}$. Using \ref{lem:lp-continuity-opt-value} and Theorem 3.1 in \cite{davidson96stability-extreme}, the extreme point set of $\mathcal{O}^*_{\mathbf{c}}$ is continuous at $\left( \bm{\mu}, \bm{\lambda} \right)$. Since $\mathcal{O}^*_{\mathbf{c}}$ is convex and is a singleton at $\left( \bm{\mu}, \bm{\lambda} \right)$, it follows that $\mathcal{O}^*_{\mathbf{c}}$ is continuous at $\left( \bm{\mu}, \bm{\lambda} \right)$.
\end{proof}

\begin{proof}[Proof of Lemma~\ref{lem:perturbed-unique-soln}]
For any $(\bm{\mu}, \bm{\lambda}) \in \mathcal{S}$, it holds that $(\bm{\mu}, \bm{\lambda}) \in \mathcal{U}_{\mathbf{c}^{\epsilon_p}}$ if the vector $\mathbf{c}^{\epsilon_p}$ is not perpendicular to any of the faces of the polytope given by the feasible set of $\mathit{L}_{\mathbf{c}^{\epsilon_p}}(\bm{\mu}, \bm{\lambda})$. For any $1 \leq k \leq d-1$, consider any $k$-dimensional face of this polytope. The probability that the vector $\mathbf{c}^{\epsilon_p}$ lies in the $d-k$ dimensional space orthogonal to this face is zero. Since there are only a finite number of faces, by the union bound, we have
$$\PP\left[(\bm{\mu}, \bm{\lambda}) \notin \mathcal{U}_{\mathbf{c}^{\epsilon_p}}  \given \mathbf{J}(0), \mathbf{Q}(1) \right] = 0.$$
\end{proof}

\section{Proof of Theorem~\ref{thm:static-split-optimal}}
We use continuity of the linear program $\mathit{L}$ (Lemma~\ref{lem:lp-continuity}) to prove  part~\ref{item:cost-opt1} of Theorem~\ref{thm:static-split-optimal}. To prove part~\ref{item:stab1} of Theorem~\ref{thm:static-split-optimal}, we show that the long term Lyapunov drift is negative.
\subsection{Cost Optimality}
Part~\ref{item:cost-opt1} of the theorem follows easily from the continuity of the optimal value of the linear program $\mathit{L}$.
Since $\mathbf{P}(\bm{\sigma}^*, \epsilon_s)$ has a single ergodic class, the marginal distribution of the Markov chain $\left\lbrace \mathbf{J}(t) \right\rbrace_{t \geq 0}$ converges to $\bm{\sigma}^*$. This gives us
\begin{align*}
& \limsup_{t \to \infty} \EE\left[ C(t) \given[\big]  \mathbf{J}(0), \mathbf{Q}(1) \right]	\\	
& \leq (1 - \epsilon_s) \sum_{j' \in \mathcal{J}} \sigma^*_{j'} \mathsf{C}_1  \norm{j'}_1 + \epsilon_s \left( M \mathsf{C}_0 + \sum_{j \in \mathcal{J}} \sigma^*_{j} \mathsf{C}_1  \norm{j}_1 \right) 	\\
 & \leq C^*_{\mathbf{c}^{0}}\left( \bm{\mu}, \bm{\lambda}+\epsilon_g \right) +  M \mathsf{C}_0 \epsilon_s	\\
 & \leq C^*_{\mathbf{c}^0}\left( \bm{\mu}, \bm{\lambda} \right) +  M \mathsf{C}_0 \epsilon_s + \gamma(\epsilon_g),
\end{align*}
 for some increasing function $\gamma(\cdot)$ such that $\lim_{\epsilon_g \to 0} \gamma(\epsilon_g) = 0$. This follows  from the continuity of $C^*_{\mathbf{c}^0}\left( \bm{\mu}, \cdot \right)$ (part~\ref{lem:lp-continuity-opt-value} of Lemma~\ref{lem:lp-continuity}). Therefore, for $\kappa =  M \mathsf{C}_0$,
\begin{align*}
\limsup_{T \to \infty} \frac{1}{T} \sum_{t=1}^T &\EE\left[ C(t) \given[\big]  \mathbf{J}(0), \mathbf{Q}(1) \right]	\\	
& \leq C^*_{\mathbf{c}^0}(\bm{\mu}, \bm{\lambda}) + \kappa \epsilon_s + \gamma(\epsilon_g)	\\
& \leq C^{\mathfrak{M}}(\bm{\mu}, \bm{\lambda}) + \kappa \epsilon_s + \gamma(\epsilon_g),
\end{align*}
where the last inequality follows from \eqref{eq:lp-opt}. This proves part~\ref{item:cost-opt1} of Theorem~\ref{thm:static-split-optimal}.

\subsection{Stability: Negative Lyapunov Drift}
We show stability in the sense of Definition~\ref{def:stability} by showing that the quadratic Lyapunov drift for the Markov policy $\varphi(\bm{\mu}, \bm{\lambda}+\epsilon_g, \epsilon_s)$ is negative outside a finite set. Let $V(\mathbf{q}) := \sum_{m,u} q_{m,u}^2$ be the Lyapunov function. For any $T > 0$, $t > 0$, let $\Delta_T(t) := V(\mathbf{Q}(t+T)) - V(\mathbf{Q}(t))$ be the $T$-step Lyapunov drift. Due to Foster's theorem, it is sufficient to prove the following lemma to prove part~\ref{item:stab1} of Theorem~\ref{thm:static-split-optimal}.
\begin{lemma}
\label{lem:static-split-lyapunov-drift}
For any $\bm{\mu}, \bm{\lambda}$, there exists constants $T$, $B$ such that for any $t\in \NN$,
\begin{align*}
& \EE_{\varphi(\bm{\mu}, \bm{\lambda}+\epsilon_g, \epsilon_s)}\left[\Delta_T(t) \given[\big]  \mathbf{J}(t-1), \mathbf{Q}(t) \right]		\\ 	
& \leq B T -  \epsilon_g T \sum_{m,u} Q_{m,u}(t).
\end{align*}
\end{lemma}
\begin{proof}
Since the marginal distribution of the Markov chain $\left\lbrace \mathbf{J}(t) \right\rbrace_{t \geq 0}$ converges to $\bm{\sigma}^*$, we can choose a constant $T \in \NN$ such that
\begin{align}
\label{eq:pick-T}
\max_{j' \in \mathcal{J}} \sum_{l=0}^{T-1} \sum_{j \in \mathcal{J}} \card[\Big]{\PP\left[ \mathbf{J}(l) = j \given[\big] \mathbf{J}(0) = j' \right] - \sigma^*_j} \leq \frac{T \epsilon_g}{2 \bar{\mathsf{R}}}.
\end{align}
Since we have bounded arrivals and service, for $B' = n M \max\{\bar{\mathsf{A}}^2, \bar{\mathsf{R}}^2\}$, the $T$-step drift satisfies
\begin{align*}
\Delta_T(t) & \leq B' T + 2 \sum_{l=0}^{T-1} ( \mathbf{A}(t+l) - \mathbf{S}(t+l) ) \cdot \mathbf{Q}(t+l).
\end{align*}
Let $\bm{\alpha}^*$ be the set of convex combinations related to the unique optimal solution  $(\bm{\sigma}^*, \bm{\beta}^*)$ through \eqref{eq:change-var}. Since the policy $\varphi(\bm{\mu}, \bm{\lambda}+\epsilon_g, \epsilon_s)$ allocates rates according to the Max-Weight rule, we have
\begin{align*}
& \mathbf{S}(t+l)  \cdot \mathbf{Q}(t+l)	\\
& = \max_{\mathbf{r} \in \mathcal{R}(\mathbf{J}(t+l),H(t+l))} \mathbf{r}  \cdot \mathbf{Q}(t+l) 	\\
& \geq \sum_{\mathbf{r} \in \mathcal{R}(\mathbf{J}(t+l),H(t+l))} \alpha^*_{\mathbf{r}}(\mathbf{J}(t+l),H(t+l)) \mathbf{r}  \cdot \mathbf{Q}(t+l),
\end{align*}
for any $l \in \{0, 1, \dots, T-1\}$.
Using the above inequality and that the arrivals and service per time-slot are bounded at every queue, for $B = B' + n M \max\{\bar{\mathsf{A}}^2, \bar{\mathsf{R}}^2\}(T-1)$, we obtain
\begin{eqnarray*}
\lefteqn{\Delta_T(t)}\\
& \leq & B T + 2 \sum_{l=0}^{T-1}  \mathbf{A}(t+l) \cdot \mathbf{Q}(t)	\\
& & -  2 \sum_{l=0}^{T-1} \sum_{\mathbf{r} \in \mathcal{R}(\mathbf{J}(t+l),H(t+l))} \hspace*{-.2in} \alpha^*_{\mathbf{r}}(\mathbf{J}(t+l),H(t+l)) \mathbf{r}  \cdot \mathbf{Q}(t).
\end{eqnarray*}
Taking averages in the above inequality, we get
\begin{align}
\EE\left[\Delta_T(t) \given[\big]  \mathbf{J}(t-1), \mathbf{Q}(t) \right]	
 & \leq B T + 2 (T \bm{\lambda} - Z) \cdot \mathbf{Q}(t) \label{eq:avg-drift},
\end{align}
where
\begin{align*}
Z := \sum_{l=0}^{T-1}  \EE\left[ \sum_{\mathbf{r} \in \mathcal{R}(\mathbf{J}(t+l),H(t+l))} \hspace*{-.4in} \alpha^*_{\mathbf{r}}(\mathbf{J}(t+l),H(t+l)) \mathbf{r}  \given[\bigg]  \mathbf{J}(t-1) \right].
\end{align*}
Now, for any $l \in [0,T-1]$, let
\begin{align*}
y(t+l)_j & := \PP\left[ \mathbf{J}(t+l) = j \given[\big] \mathbf{J}(t-1) \right].
 \end{align*}
 Then,
\begin{align*}
& \sum_{l=0}^{T-1}  \EE\left[ \sum_{\mathbf{r} \in \mathcal{R}(\mathbf{J}(t+l),H(t+l))} \hspace*{-.1in} \alpha^*_{\mathbf{r}}(\mathbf{J}(t+l),H(t+l)) \mathbf{r}  \given[\bigg]  \mathbf{J}(t-1) \right]	\\
& = \sum_{l=0}^{T-1}  \sum_{j \in \mathcal{J}} y(t+l)_j  \left( \sum_{h \in \mathcal{H}} \mu(h) \sum_{\mathbf{r} \in \mathcal{R}(j,h)} \alpha^*_{\mathbf{r}}(j,h) \mathbf{r} \right)	\\
& \geq T   \left( \sum_{j \in \mathcal{J}} \sigma^*_j \sum_{h \in \mathcal{H}} \mu(h) \sum_{\mathbf{r} \in \mathcal{R}(j,h)}  \alpha^*_{\mathbf{r}}(j,h) \mathbf{r} \right)	\\
& \quad - \sum_{l=0}^{T-1} \norm{\mathbf{y}(t+l) - \bm{\sigma}^*}_1 \bar{\mathsf{R}}	\\
& \geq T (\bm{\lambda} + \epsilon_g/2 ),
\end{align*}
where the last inequality follows from \eqref{eq:pick-T} and the fact that any solution to the linear program $\mathit{L}_{\mathbf{c}^{0}}\left( \bm{\mu}, \bm{\lambda}+\epsilon_g \right)$ satisfies its constraints
\begin{align*}
\sum_{j \in \mathcal{J}} \sigma^*_j \sum_{h \in \mathcal{H}} \mu(h) \sum_{\mathbf{r} \in \mathcal{R}(j,h)} \alpha^*_{\mathbf{r}}(j,h) \mathbf{r} \geq \bm{\lambda} + \epsilon_g.
\end{align*}
Substituting this inequality in \eqref{eq:avg-drift}, we get the required result
\begin{align*}
\EE\left[\Delta_T(t)  \given[\big]  \mathbf{J}(t-1), \mathbf{Q}(t) \right]	
& \leq B T -  T \epsilon_g \sum_{m,u} Q_{m,u}(t).
\end{align*}
\end{proof}

\section{Proof of Theorem~\ref{thm:algo-optimal}}
As in the proof of Theorem~\ref{thm:static-split-optimal}, we use continuity of the linear program $\mathit{L}$ (Lemma~\ref{lem:lp-continuity}) to prove  part~\ref{item:cost-opt2} of Theorem~\ref{thm:algo-optimal}. To prove stability (part~\ref{item:stab2} of Theorem~\ref{thm:algo-optimal}), we show that the long term Lyapunov drift is negative outside a finite set given the event
\begin{align}
\label{eq:unique-soln-event}
\mathcal{E}^0 := \left( \bm{\mu}, \bm{\lambda}+\epsilon_g \right) \in \mathcal{U}_{\mathbf{c}^{\epsilon_p}}.
\end{align}
This negative Lyapunov drift, as in the Foster's theorem for time-homogeneous Markov chains, is then used to prove stability as per Definition~\ref{def:stability}.

\subsection{Cost Optimality}
According to the BS activation rule of policy $\phi(\epsilon_p, \epsilon_s, \epsilon_g)$, we have $\mathbf{J}(t) = (1-E_l(t))\tilde{\mathbf{J}}(t) + E_l(t) \mathbf{1}$ and $\mathbf{J}(t) \geq \tilde{\mathbf{J}}(t),$ which in turn implies that
\begin{align*}
& \norm{\left( \mathbf{J}(t-1) - \mathbf{J}(t) \right)^+}_1 	\\
& \leq M E_l(t-1) + \norm{\left( \mathbf{J}(t-1) - \mathbf{J}(t)  \right)^+}_1 (1-E_l(t-1))	\\
& \leq M E_l(t-1) + \norm{\left( \tilde{\mathbf{J}}(t-1) - \tilde{\mathbf{J}}(t) \right)^+}_1,
\end{align*}
where the last inequality can be checked via an straightforward case-by-case analysis of $E_l(t-1)$ and $E_l(t)$. Let
$$\mathbf{z}(t) := \left( \mathds{1}\left\lbrace \tilde{\mathbf{J}}(t) = j \right\rbrace \right)_{j \in \mathcal{J}},$$
$$\mathbf{y}(t) := \left( \mathds{1}\left\lbrace \mathbf{J}(t) = j \right\rbrace \right)_{j \in \mathcal{J}}.$$
Then, the expected cost at time $t$ under policy $\phi(\epsilon_p, \epsilon_s, \epsilon_g)$ is given by
\begin{align}
\label{eq:policy-cost}
& \EE\left[ C(t) \given[\big]  \mathbf{J}(0), \mathbf{Q}(1) \right]	\n	\\
& = \EE\left[ \mathsf{C}_0 \norm{\left( \mathbf{J}(t-1) - \mathbf{J}(t) \right)^+}_1 + \mathsf{C}_1 \norm{\mathbf{J}(t)}_1 \given[\Big]  \mathbf{J}(0) \right] 	\n	\\
& \leq \mathsf{C}_0 \left( M \epsilon_l(t-1) + \EE\left[ \norm{\left( \tilde{\mathbf{J}}(t-1) - \tilde{\mathbf{J}}(t) \right)^+}_1 \given[\Big]  \mathbf{J}(0) \right] \right) 	\n	\\
& \quad + \mathsf{C}_1 \EE\left[ \norm{\mathbf{J}(t)}_1 \given[\Big]  \mathbf{J}(0) \right]	\n	\\	
& = \mathsf{C}_0 M \epsilon_l(t-1)	\n	\\
& \quad + \mathsf{C}_0 \sum_{j', j \in \mathcal{J}} \EE\left[ z(t-1)_{j'} z(t)_{j}  \given[\big]  \mathbf{J}(0) \right] \norm{\left( j' - j \right)^+}_1 	\n	\\
& \quad + \mathsf{C}_1  \sum_{j \in \mathcal{J}} \EE[y(t)_j \given[\big]  \mathbf{J}(0)] \norm{j}_1.
\end{align}
In the rest of the proof, we will suppress the dependence on the initial state $\mathbf{J}(0)$ for convenience of notation.

Let $( \hat{\bm{\mu}}(t), \hat{\bm{\lambda}}(t) )$ be the estimated parameters at the beginning of time-slot $t$. Observe that $\hat{\bm{\mu}}(\cdot)$, $\hat{\bm{\lambda}}(\cdot)$, and consequently $\hat{\bm{\sigma}}(\cdot)$ and $\hat{\bm{\beta}}(\cdot)$, are modified only at times $t$ when $E_l(t)=1$.
Now, consider a sample path that fixes $(E_l(\cdot),  \hat{\bm{\mu}}(\cdot), \hat{\bm{\lambda}}(\cdot))$. Conditioned on this sample path, the process ${\bm{z}}(\cdot)$ is a time-inhomogeneous Markov chain with transition probability matrix $\mathbf{P}(\hat{\bm{\sigma}}(l), \epsilon_s)$ at time $l$. Hence
\begin{align*}
& \EE\left[ \mathbf{z}(t) \given[\big] (E_l(\cdot), \hat{\bm{\mu}}(\cdot), \hat{\bm{\lambda}}(\cdot)) \right]	\\
& = \EE\left[ \mathbf{z}(0) \prod_{l=1}^t \mathbf{P}(\hat{\bm{\sigma}}(l), \epsilon_s) \given[\big] (E_l(\cdot), \hat{\bm{\mu}}(\cdot), \hat{\bm{\lambda}}(\cdot)) \right],
\end{align*}
which when unconditioned yields
$$
\EE\left[ \mathbf{z}(t) \right] = \EE\left[ \mathbf{z}(0) \prod_{l=1}^t \mathbf{P}(\hat{\bm{\sigma}}(l), \epsilon_s) \right].
$$
Given $\mathcal{E}^0$ as defined in \eqref{eq:unique-soln-event}, let  $(\bm{\sigma}^*, \bm{\beta}^*) \in \mathcal{O}^*_{\mathbf{c}^{\epsilon_p}}\left( \bm{\mu}, \bm{\lambda}+\epsilon_g \right)$ be the unique solution to the linear program $\mathit{L}_{\mathbf{c}^{\epsilon_p}}\left( \bm{\mu}, \bm{\lambda}+\epsilon_g \right)$.
Since  $\lim_{t \to \infty} \sum_{s=1}^{t} E_l(s) \overset{\text{a.s.}}{=} \infty$, we have $\lim_{t \to \infty} \bigl( \hat{\bm{\mu}}(t), \hat{\bm{\lambda}}(t) \bigr) \overset{\text{a.s.}}{=} \left( \bm{\mu}, \bm{\lambda} \right)$ and from part~\ref{lem:lp-continuity-opt-set} of Lemma~\ref{lem:lp-continuity}, $\lim_{t \to \infty} \hat{\bm{\sigma}}(t) \overset{\text{a.s.}}{=} \bm{\sigma}^*$ and $\lim_{t \to \infty} \mathbf{P}(\hat{\bm{\sigma}}(t), \epsilon_s) \overset{\text{a.s.}}{=} \mathbf{P}(\bm{\sigma}^*, \epsilon_s)$. Furthermore, using Lemma~\ref{lem:inhomo-dtmc-convergence}\ref{item:limiting-bound} and the limit law under a sample path $(E_l(\cdot), \hat{\bm{\mu}}(\cdot), \hat{\bm{\lambda}}(\cdot))$ with all these properties, we also have
$$\lim_{t \to \infty} \mathbf{z}(0) \prod_{l=1}^t \mathbf{P}(\hat{\bm{\sigma}}(l), \epsilon_s) \overset{\text{a.s.}}{=} \bm{\sigma}^*,$$
which gives us $\lim_{t \to \infty} \EE\left[ \mathbf{z}(t) \right] = \bm{\sigma}^*$ by the bounded convergence theorem, and
$$\lim_{t \to \infty} \EE[\mathbf{y}(t)] = \lim_{t \to \infty} (1-\epsilon_l(t)) \EE\left[ \mathbf{z}(t) \right] + \epsilon_l(t) \bm{\eta} = \bm{\sigma}^*.$$
Similarly, for any $j', j \in \mathcal{J}$,
\begin{align*}
\lim_{t \to \infty} \EE\left[ z(t-1)_{j'} z(t)_j \right] & = \lim_{t \to \infty} \EE\left[ z(t-1)_{j'} \mathbf{P}(\hat{\bm{\sigma}}(t), \epsilon_s)_{j',j} \right]	\\
 & = \EE\left[ \lim_{t \to \infty} z(t-1)_{j'} \mathbf{P}(\hat{\bm{\sigma}}(t), \epsilon_s)_{j',j} \right]	\\
 & = \lim_{t \to \infty} \EE\left[ z(t-1)_{j'} \right] \mathbf{P}(\bm{\sigma}^*, \epsilon_s)_{j',j}	\\
 & = \sigma^*_{j'} \mathbf{P}(\bm{\sigma}^*, \epsilon_s)_{j',j},
\end{align*}
where the second equality is once again due to the bounded convergence theorem. Applying these along with Lemma~\ref{lem:perturbed-unique-soln} to \eqref{eq:policy-cost} yields
\begin{align*}
& \limsup_{t \to \infty} \EE\left[ C(t) \given[\big]  \mathbf{J}(0), \mathbf{Q}(1) \right]	\\	
& \leq \limsup_{t \to \infty} \mathsf{C}_0 \sum_{j', j \in \mathcal{J}} \EE\left[ z(t-1)_{j'} z(t)_j \right]
 \norm{\left( j' - j \right)^+}_1	\\
 & \quad + \limsup_{t \to \infty} \mathsf{C}_1  \sum_{j \in \mathcal{J}} \EE[y(t)_j] \norm{j}_1	\\
& = \mathsf{C}_0 \sum_{j', j \in \mathcal{J}} \sigma^*_{j'} \mathbf{P}(\bm{\sigma}^*, \epsilon_s)_{j',j} \norm{\left( j' - j \right)^+}_1 + \mathsf{C}_1 \sum_{j \in \mathcal{J}} \sigma^*_{j} \norm{j}_1	\\
 & \leq C^*_{\mathbf{c}^{\epsilon_p}}\left( \bm{\mu}, \bm{\lambda}+\epsilon_g \right) +  M \mathsf{C}_0 \epsilon_s	\\
 & \leq C^*_{\mathbf{c}^0}\left( \bm{\mu}, \bm{\lambda}+\epsilon_g \right) + \sqrt{\card{\mathcal{H}}+1}  \mathsf{C}_1 \epsilon_p +  M \mathsf{C}_0 \epsilon_s	\\
 & \leq C^*_{\mathbf{c}^0}\left( \bm{\mu}, \bm{\lambda} \right) + \kappa(\epsilon_p + \epsilon_s) + \gamma(\epsilon_g),
\end{align*}
 for some increasing function $\gamma(\cdot)$ such that $\lim_{\epsilon_g \to 0} \gamma(\epsilon_g) = 0$ and $\kappa = \max(\sqrt{\card{\mathcal{H}}+1}  \mathsf{C}_1, M \mathsf{C}_0)$.
 This gives us
\begin{align*}
& \limsup_{T \to \infty} \frac{1}{T} \sum_{t=1}^T \EE\left[ C(t) \given[\big]  \mathbf{J}(0), \mathbf{Q}(1) \right]	\\
& \leq C^*_{\mathbf{c}^0}(\bm{\mu}, \bm{\lambda}) + \kappa (\epsilon_p + \epsilon_s) + \gamma(\epsilon_g)	\\
& \leq C^{\mathfrak{M}}(\bm{\mu}, \bm{\lambda}) + \kappa (\epsilon_p + \epsilon_s) + \gamma(\epsilon_g),
\end{align*}
where the last inequality follows from \eqref{eq:lp-opt}. This proves part~\ref{item:cost-opt2} of Theorem~\ref{thm:algo-optimal}.

\subsection{Stability: Negative Lyapunov Drift}
Similar to Theorem~\ref{thm:static-split-optimal}, we show in Lemma~\ref{lem:lyapunov-drift} that the long term Lyapunov drift is negative outside a finite set. But unlike in Theorem~\ref{thm:static-split-optimal}, this negative drift condition holds only after a random time that has bounded second moment.
\begin{lemma}
\label{lem:lyapunov-drift}
For any $\bm{\mu}, \bm{\lambda}$, there exists constants $T$, $B$ and a random time $\Gamma$ such that $\EE\left[ \Gamma^2 \given \mathbf{J}(0), \mathbf{Q}(1) \right] < \infty$, and for any $t > \Gamma$,
\begin{align}
\label{eq:lyapunov-drift}
& \EE_{\phi(\epsilon_p, \epsilon_s, \epsilon_g)}\left[\Delta_T(t) \given[\big] \tilde{\mathbf{J}}(t-1), \mathbf{Q}(t), \mathbf{J}(0), \mathbf{Q}(1), \Gamma, \mathcal{E}^0 \right] 	\n	\\	
& \leq B T -  \epsilon_g T \sum_{m,u} Q_{m,u}(t).
\end{align}
\end{lemma}
Before we prove Lemma~\ref{lem:lyapunov-drift}, we show below that \eqref{eq:lyapunov-drift} implies stability as per Definition~\ref{def:stability}.
\begin{lemma}
\label{lem:fosters}
If the condition given by \eqref{eq:lyapunov-drift} is satisfied, then for any $b > 0$,
\begin{align*}
\lim_{t \to \infty} \frac{1}{t} \sum_{l=1}^{t}\PP_{\phi(\epsilon_p, \epsilon_s, \epsilon_g)}\left[ \mathbf{Q}(l) \in \mathcal{A} \given[\bigg]  \mathbf{J}(0), \mathbf{Q}(1) \right] > \frac{b}{\bar{B}},
\end{align*}
where $\bar{B} = B +b$, $\bar{Q} = \frac{\bar{B}}{\epsilon_g}$ and  $\mathcal{A} = \left\lbrace \mathbf{Q} \in \RR^{M \times n} :  \sum_{m,u} Q_{m,u} < \bar{Q} \right\rbrace.$  Therefore, the network is stable under the policy $\phi(\epsilon_p, \epsilon_s, \epsilon_g)$.
\end{lemma}
\begin{proof}
For ease of notation, we do not explicitly write the conditioning on $\mathbf{J}(0), \mathbf{Q}(1)$. Let the condition given by \eqref{eq:lyapunov-drift} be true. Then under policy $\phi(\epsilon_p, \epsilon_s, \epsilon_g)$, we have from \eqref{eq:lyapunov-drift} that
\begin{align*}
\EE\left[\Delta_T(t) \given[\big] \tilde{\mathbf{J}}(t-1), \mathbf{Q}(t), \Gamma, \mathcal{E}^0 \right]	
& \leq -b T + \bar{B} T \ind{\mathbf{Q}(t) \in \mathcal{A}},
\end{align*}
for any  $t > \Gamma$. Now, let $I^* := \min\{i: (i-1)T \geq \Gamma\}$.  Consider, for any $k \in \NN$,
\begin{align}
& \sum_{l=1}^{T} \EE\left[ V(\mathbf{Q}(kT+l)) - V(\mathbf{Q}(l))  \right]	\n	\\
& = \sum_{l=1}^{T} \EE\left[ V(\mathbf{Q}(kT+l) -  V(\mathbf{Q}((I^*-1)T+l))  \right]	\n	\\
& \quad +  \EE\left[ V(\mathbf{Q}((I^*-1)T+l)) - V(\mathbf{Q}(l))  \right].	\label{eq:fosters}
\end{align}
Now,
\begin{align*}
& \sum_{l=1}^{T} \EE\left[ V(\mathbf{Q}(kT+l)) -  V(\mathbf{Q}((I^*-1)T+l))  \right]	\\
& = \sum_{l=1}^{T} \EE\left[ \sum_{i=I^*-1}^{k-1}  \Delta_T(iT+l) \right]	\\
& = \sum_{l=1}^{T} \EE\left[ \sum_{i=I^*-1}^{k-1} \EE\left[ \Delta_T(iT+l) \given \mathbf{Q}(iT+l), \Gamma, \mathcal{E}^0 \right] \right] \label{eq:smooth}	\numberthis		\\
& \leq  \sum_{l=1}^{T} \EE\left[  \sum_{i=I^*-1}^{k-1} \left(  -b T + \bar{B} T \ind{\mathbf{Q}(iT+l) \in \mathcal{A}} \right)  \right]	\\
& \leq -b (k-1)T^2 + b T \EE\left[ \Gamma  \right]  + \bar{B} T \EE\left[ \sum_{t=\Gamma+1}^{kT} \ind{\mathbf{Q}(t) \in \mathcal{A}}  \right].
\end{align*}
In \eqref{eq:smooth}, we have used  that $\PP\left[ \mathcal{E}^0  \right] = 1$ from Lemma~\ref{lem:perturbed-unique-soln}.

Moreover, since $(I^*-1)T < \Gamma + T$, we have $$\mathbf{Q}((I^*-1)T+l) \leq \mathbf{Q}(l) + \bar{\mathsf{A}}(\Gamma + T)$$ for any $l \in \NN$, which gives
\begin{align*}
& \sum_{l=1}^{T} \EE\left[ V(\mathbf{Q}((I^*-1)T+l)) - V(\mathbf{Q}(l)) \right]	\\
 & \leq \sum_{l=1}^{T} \EE\left[ \sum_{m,u} \left(  (\bar{\mathsf{A}}(\Gamma + T) +  Q_{m,u}(l))^2 - Q_{m,u}(l)^2  \right) \right]	\\
 & \leq \EE\left[ n M T (\bar{\mathsf{A}}(\Gamma + T))^2 + 2 \sum_{l=1}^{T} \sum_{m,u}  \bar{\mathsf{A}}(\Gamma + T) Q_{m,u}(l) \right]	\\
 & \leq  \EE\left[ n M T (\bar{\mathsf{A}}(\Gamma + T))^2 \right]	\\
 & \quad + \EE\left[ 2 \bar{\mathsf{A}}(\Gamma + T) \sum_{l=1}^{T} \sum_{m,u}  \left( \bar{\mathsf{A}}(l-1) + Q_{m,u}(1) \right) \right] 	\\
 & \leq T \EE\left[ n M \bar{\mathsf{A}}^2 (\Gamma^2 + 3T \Gamma + 2T^2) \right]	\\
 & \quad + T \EE\left[  2 \bar{\mathsf{A}} (\Gamma + T) \sum_{m,u} Q_{m,u}(1) \right].
\end{align*}
Applying the above two inequalities in \eqref{eq:fosters} and rearranging the terms, we get
\begin{align*}
& \bar{B} \EE\left[ \sum_{t=\Gamma+1}^{kT} \ind{\mathbf{Q}(t) \in \mathcal{A}} \right] 	\\
& \geq bkT - \EE\left[Y \right]	\\
& \quad + \frac{1}{T}\sum_{l=1}^{T} \EE\left[ V(\mathbf{Q}(kT+l)) - V(\mathbf{Q}(l)) \right],
\end{align*}
where
\begin{align*}
Y & = n M \bar{\mathsf{A}}^2 \Gamma^2 + \left( 3 n M \bar{\mathsf{A}}^2T + 2 \bar{\mathsf{A}} \sum_{m,u} Q_{m,u}(1) + b \right) \Gamma 	\\
& \quad + 2 T^2 + 2 \bar{\mathsf{A}} T \sum_{m,u} Q_{m,u}(1) + bT.
\end{align*}
We have $\EE\left[ Y \right] < \infty$ since $\EE\left[ \Gamma^2 \right] < \infty$.
Additionally, we have
\begin{align*}
& \limsup_{k \to \infty} \frac{1}{kT} \EE\left[ \frac{1}{T}\sum_{l=1}^{T} V(\mathbf{Q}(l))  \right]	\\
& \leq \limsup_{k \to \infty} \frac{1}{kT} \EE\left[ \frac{1}{T}\sum_{l=1}^{T} V(\bar{\mathsf{A}}(l-1) + \mathbf{Q}(1))   \right] = 0.
\end{align*}
Therefore,
\begin{align*}
& \bar{B} \liminf_{k \to \infty} \frac{1}{kT} \sum_{t=1}^{kT} \PP\left[ \mathbf{Q}(t) \in \mathcal{A}   \right] \\
 & \geq \bar{B} \liminf_{k \to \infty} \frac{1}{kT} \EE\left[ \sum_{t=\Gamma+1}^{kT} \ind{\mathbf{Q}(t) \in \mathcal{A}}  \right]	\\
 & \geq b -  \limsup_{k \to \infty} \frac{1}{kT} \EE\left[ Y + \frac{1}{T}\sum_{l=1}^{T} V(\mathbf{Q}(l))   \right],
\end{align*}
which gives us the required result
\begin{align*}
\liminf_{k \to \infty} \frac{1}{kT} \sum_{t=1}^{kT} \PP\left[ \mathbf{Q}(t) \in \mathcal{A} \given[\big] \mathbf{J}(0),  \mathbf{Q}(1)  \right] \geq \frac{b}{\bar{B}}.
\end{align*}
\end{proof}
Before we proceed to prove Lemma~\ref{lem:lyapunov-drift}, we will prove an intermediate result which shows that the transition probability matrices used by the policy to select the activation vector converge to a matrix that is close to the optimal. This result along with Lemma~\ref{lem:inhomo-dtmc-convergence} allows us to show that the distribution of the activation vector converges to the optimal invariant distribution.

For any $k > 0$, let $\tilde{\bm{\mu}}(k)$, $\tilde{\bm{\lambda}}(k)$ denote the empirical distributions of channels and the empirical means of arrivals respectively obtained from the first $k$ explore samples.  For every $t > 0$, $\delta > 0$, define the events
\begin{align*}
\mathcal{E}_l(t) & := \left\lbrace \sum_{s=1}^{t} E_l(s) \geq \frac{1}{2}\log^2t \right\rbrace,	\\
\mathcal{E}_{\tilde{\mu}}(t, \delta) & := \left\lbrace \norm{\tilde{\bm{\mu}}(t) - \bm{\mu}}_1 \leq \delta \right\rbrace,	\\
\mathcal{E}_{\tilde{\lambda}}(t, \delta) & := \left\lbrace \norm{\tilde{\bm{\lambda}}(t) - \bm{\lambda}}_1 \leq\delta \right\rbrace,
\end{align*}
and for $\delta' = \frac{1}{2} \min (\delta, \epsilon_g/\bar{\mathsf{R}})$,
\begin{align*}
\mathcal{E}(t, \delta)  & := \mathcal{E}_l(t) \bigcap_{k \geq \frac{1}{2}\log^2t}\left\lbrace  \mathcal{E}_{\tilde{\mu}}(k, \delta') \cap \mathcal{E}_{\tilde{\lambda}}(k, \delta')\right\rbrace.
\end{align*}

\begin{lemma}
\label{lem:est-convergence}
For any $\delta > 0$,  let $$T_1(\delta) := \min\left\lbrace t: \mathcal{E}(t, \delta) \text{ is true} \right\rbrace.$$ Then, $\EE\left[ T_1(\delta)^2 \given \mathbf{J}(0), \mathbf{Q}(1)  \right]  < \infty.$
\end{lemma}
\begin{proof}
For ease of notation, we do not explicitly write the conditioning on $\mathbf{J}(0), \mathbf{Q}(1)$.

Consider the mean number of explore samples in the first $t$ slots.
\begin{align*}
\EE\left[ \sum_{s=1}^{t} E_l(s) \right] & =  \sum_{s=2}^{t} \frac{2 \log s}{s} \geq \int_{s=e}^{t+1} \frac{2 \log s}{s} \diff s	\\
& = \log^2(t+1) - 1  \geq \frac{3}{4} \log^2(t),
\end{align*}
for all $t \geq 3$. Using the Chernoff bound for Bernoulli random variables, we have $\forall t \geq 3$,
\begin{align*}
\PP\left[ \mathcal{E}_l(t)^c \right] \leq \exp\left( -\frac{1}{32}\log^2t \right).
\end{align*}
Using the Hoeffding's inequality for Bernoulli random variables, for any  $k$, $\epsilon >0$,
\begin{align*}
\PP\left[ \mathcal{E}_{\tilde{\lambda}}(k, \epsilon)^c \right] & \leq \sum_{u=1}^n \PP\left[ \card{\tilde{\lambda}_{u}(k) - \lambda_{u}} \geq \frac{1}{n}\epsilon \right] 	\\
& \leq n \exp\left( -2 k \left( \frac{\epsilon}{n \bar{\mathsf{A}}} \right)^2 \right).
\end{align*}
In addition, using Pinsker's inequality it can be shown \cite{weissman2003inequalities} that for any $k$, $\epsilon >0$,
\begin{align*}
\PP\left[ \mathcal{E}_{\tilde{\mu}}(k, \epsilon)^c \right] \leq (k+1)^{\card{\mathcal{H}}} \exp\left( -\frac{\epsilon^2}{2} k \right).
\end{align*}
Now, let $\delta' = \frac{1}{2}\min(\delta, \epsilon_g/\bar{\mathsf{R}})$. Using the above inequalities, we have $\forall t \geq 3$,
\begin{align*}
 \PP\left[ T_1(\delta) > t \right]	
& \leq \PP\left[ \mathcal{E}(t, \delta)^c \right]	\\
& \leq \PP\left[ \mathcal{E}_l(t)^c \right]	\\
& \quad +  \sum_{k=\frac{1}{2}\log^2t }^{\infty} \left( \PP\left[ \mathcal{E}_{\tilde{\mu}}(k, \delta')^c \right] + \PP\left[ \mathcal{E}_{\tilde{\lambda}}(k, \delta')^c \right] \right)	\\
& = o\left( \frac{1}{t^3} \right),
\end{align*}
which gives us
\begin{align*}
\EE\left[ T_1(\delta)^2 \right] = 2 \sum_{t = 0}^ \infty t \PP\left[ T_1(\delta) > t \right]  < \infty.
\end{align*}
\end{proof}

Given $\mathcal{E}^0$, let  $(\bm{\sigma}^*, \bm{\beta}^*) \in \mathcal{O}^*_{\mathbf{c}^{\epsilon_p}}\left( \bm{\mu}, \bm{\lambda}+\epsilon_g \right)$ be the unique solution to the linear program $\mathit{L}_{\mathbf{c}^{\epsilon_p}}\left( \bm{\mu}, \bm{\lambda}+\epsilon_g \right)$. Due to the continuity of the solution set of $\mathit{L}_{\mathbf{c}^{\epsilon_p}}\left( \bm{\mu}, \bm{\lambda}+\epsilon_g \right)$ given $\mathcal{E}^0$ (from Lemma~\ref{lem:lp-continuity}), there exists a positive function $\delta_1 \mapsto f(\delta_1)$ such that, if $\left( \bm{\mu}', \bm{\lambda}'+\epsilon_g \right) \in \mathcal{S}$ and
$$\norm{\bm{\mu}' - \bm{\mu}}_1 + \norm{\bm{\lambda}' - \bm{\lambda}}_1 \leq f(\delta_1),$$
 then for any $(\bm{\sigma}', \bm{\beta}') \in \mathcal{O}^*_{\mathbf{c}^{\epsilon_p}}\left( \bm{\mu}', \bm{\lambda}'+\epsilon_g \right)$,
 $$\norm{\bm{\sigma}' - \bm{\sigma}^*}_1 \leq \delta_1.$$
 The following lemma shows that continuity of the solution of  $\mathit{L}_{\mathbf{c}^{\epsilon_p}}\left( \bm{\mu}, \bm{\lambda}+\epsilon_g \right)$ implies convergence of the activation vector transition probability matrices.
 \begin{lemma}
 \label{lem:P-mx-convergence}
If $(\bm{\mu}, \bm{\lambda}+2\epsilon_g) \in \mathcal{S}$, then for any $\delta_1$, $t$, the event $\mathcal{E}^0 \cap \mathcal{E}(t, f(\delta_1))$ implies the event
  $$\tilde{\mathcal{E}}(t, \delta_1) := \left\lbrace\norm{\mathbf{P}(\hat{\bm{\sigma}}(l), \epsilon_s) - \mathbf{P}(\bm{\sigma}^*, \epsilon_s)}_1 \leq \delta_1 \; \forall l > t\right\rbrace.$$
\end{lemma}
\begin{proof}
The event $\mathcal{E}(t,  f(\delta_1))$ implies that for any $l > t$, $(\hat{\bm{\mu}}(l), \hat{\bm{\lambda}}(l)+\epsilon_g) \in \mathcal{S}$ and
$$\norm{\hat{\bm{\mu}}(l) - \bm{\mu}}_1 + \norm{\hat{\bm{\lambda}}(l) - \bm{\lambda}}_1 \leq  f(\delta_1).$$
 Therefore, we have
\begin{align*}
\norm{\hat{\bm{\sigma}}(l) - \bm{\sigma}^*}_1 \leq \delta_1,
\end{align*}
which gives us
\begin{align*}
\norm{\mathbf{P}(\hat{\bm{\sigma}}(l), \epsilon_s)  - \mathbf{P}(\bm{\sigma}^*, \epsilon_s)}_1
 \leq \norm{\hat{\bm{\sigma}}(l) - \bm{\sigma}^*}_1 \leq \delta_1.
\end{align*}
\end{proof}


We now prove the negative Lyapunov drift condition (Lemma~\ref{lem:lyapunov-drift}).
\begin{proof}[Proof of Lemma~\ref{lem:lyapunov-drift}]
\begin{figure*}[!t]
\normalsize
\begin{align}
Z & := \sum_{l=0}^{T-1}  \EE\left[ \sum_{\mathbf{r} \in \mathcal{R}(\tilde{\mathbf{J}}(t+l),H(t+l))} \alpha^*_{\mathbf{r}}(\tilde{\mathbf{J}}(t+l),H(t+l)) \mathbf{r} \given[\bigg]  \tilde{\mathbf{J}}(t-1), \mathbf{Q}(t), \mathbf{J}(0), \mathbf{Q}(1), \Gamma, \mathcal{E}^0  \right]	\label{eq:dummy-var}	\\
& = \sum_{l=0}^{T-1}  \EE\left[ \EE\left[ \sum_{\mathbf{r} \in \mathcal{R}(\tilde{\mathbf{J}}(t+l),H(t+l))} \alpha^*_{\mathbf{r}}(\tilde{\mathbf{J}}(t+l),H(t+l)) \mathbf{r} \given[\bigg]  Y(t), \mathcal{E}^0  \right]  \given[\bigg]  \tilde{\mathbf{J}}(t-1), \Gamma, \mathcal{E}^0  \right]	\n	\\
& = \sum_{l=0}^{T-1}  \EE\left[ \sum_{j \in \mathcal{J}} y(t+l)_j \left( \sum_{h \in \mathcal{H}} \mu(h) \sum_{\mathbf{r} \in \mathcal{R}(j,h)} \alpha_{\mathbf{r}}(j,h) \mathbf{r} \right) \given[\bigg]  \tilde{\mathbf{J}}(t-1), \Gamma, \mathcal{E}^0  \right]	\n	\\
& \geq \sum_{l=0}^{T-1}  \EE\left[ \left( \sum_{j \in \mathcal{J}} \sigma^*_j \sum_{h \in \mathcal{H}} \mu(h) \sum_{\mathbf{r} \in \mathcal{R}(j,h)}  \alpha_{\mathbf{r}}(j,h) \mathbf{r} \right) - \norm{\mathbf{y}(t+l) - \bm{\sigma}^*}_1 \bar{\mathsf{R}} \given[\bigg]  \tilde{\mathbf{J}}(t-1), \Gamma, \mathcal{E}^0  \right]. \label{eq:service-lb}
\end{align}

\hrulefill
\vspace*{1pt}
\end{figure*}

Fix constants $\delta_1 \in (0,1)$ and $T \in \NN$ such that 
$$\frac{2+ T \delta_1}{\epsilon_s} \leq \frac{T \epsilon_g}{2 \bar{\mathsf{R}}}.$$
Define the random time $\Gamma := T_1(f(\delta_1)).$
Using the fact that the policy $\phi(\epsilon_p, \epsilon_s, \epsilon_g)$ allocates rates according to the Max-Weight rule and that $\mathbf{J}(t+l) \geq \tilde{\mathbf{J}}(t+l)$, we have
\begin{align*}
& \mathbf{S}(t+l)  \cdot \mathbf{Q}(t+l)	\\
& = \max_{\mathbf{r} \in \mathcal{R}(\mathbf{J}(t+l),H(t+l))} \mathbf{r}  \cdot \mathbf{Q}(t+l) 	\\
& \geq \max_{\mathbf{r} \in \mathcal{R}(\tilde{\mathbf{J}}(t+l),H(t+l))} \mathbf{r}  \cdot \mathbf{Q}(t+l) 	\\
& \geq \sum_{\mathbf{r} \in \mathcal{R}(\tilde{\mathbf{J}}(t+l),H(t+l))} \alpha^*_{\mathbf{r}}(\tilde{\mathbf{J}}(t+l),H(t+l)) \mathbf{r}  \cdot \mathbf{Q}(t+l),
\end{align*}
for any $l \in \{0, 1, \dots, T-1\}$.
Following the same argument in the proof of Theorem~\ref{thm:static-split-optimal}, i.e., using the above inequality and that the arrivals and service per time-slot are bounded at every queue, for $B = B' + n M \max\{\bar{\mathsf{A}}^2, \bar{\mathsf{R}}^2\}(T-1)$, we obtain
\begin{eqnarray*}
\lefteqn{\Delta_T(t)}\\
& \leq & B T + 2 \sum_{l=0}^{T-1}  \mathbf{A}(t+l) \cdot \mathbf{Q}(t)	\\
& & -  2 \sum_{l=0}^{T-1} \sum_{\mathbf{r} \in \mathcal{R}(\tilde{\mathbf{J}}(t+l),H(t+l))} \hspace*{-.2in} \alpha^*_{\mathbf{r}}(\tilde{\mathbf{J}}(t+l),H(t+l)) \mathbf{r}  \cdot \mathbf{Q}(t).
\end{eqnarray*}
Taking averages in the above inequality, we have for any $t > \Gamma$,
\begin{align*}
& \EE\left[\Delta_T(t)  \given[\big]  \tilde{\mathbf{J}}(t-1), \mathbf{Q}(t), \mathbf{J}(0), \mathbf{Q}(1), \Gamma, \mathcal{E}^0  \right]	\\
 & \leq B T + 2 (T \bm{\lambda} - Z) \cdot \mathbf{Q}(t),
\end{align*}
where $Z$ is given by \eqref{eq:dummy-var}.

 To prove \eqref{eq:lyapunov-drift}, it is sufficient to prove that  $\EE\left[ \Gamma^2 \given \mathbf{J}(0), \mathbf{Q}(1) \right]  < \infty$, and for any $t > \Gamma$,
\begin{align}
\label{eq:excess-service}
Z \geq T (\bm{\lambda} + \epsilon_g/2 ).
\end{align}
From Lemma~\ref{lem:est-convergence}, we have that
$\EE\left[ \Gamma^2 \given \mathbf{J}(0), \mathbf{Q}(1) \right]  < \infty$.

Now to prove \eqref{eq:excess-service}, let $Y(t) = (\tilde{\mathbf{J}}(t-1), \hat{\bm{\mu}}(t), \hat{\bm{\lambda}}(t), \Gamma)$, and for any $l \in [0,T-1]$,
\begin{align*}
y(t+l)_j & := \PP\left[ \tilde{\mathbf{J}}(t+l) = j \given[\big] Y(t), \mathcal{E}^0 \right].
 \end{align*}
From Lemma~\ref{lem:P-mx-convergence}, given $\mathcal{E}^0$,   for any $t > \Gamma$ and $l \in [0,T-1]$, we have
$$\norm{\mathbf{P}(\hat{\bm{\sigma}}(t+l), \epsilon_s)  - \mathbf{P}(\bm{\sigma}^*, \epsilon_s)}_1 \leq \delta_1.$$
Recall that
$$\mathbf{P}(\bm{\sigma}^*, \epsilon_s) = \epsilon_s \mathbf{1}_{\card{\mathcal{J}}} \bm{\sigma}^* + (1-\epsilon_s) \mathbf{I}_{\card{\mathcal{J}}},$$
 and for any $l \in \NN$,
 $$\mathbf{P}(\bm{\sigma}^*, \epsilon_s)^l = \left( 1-(1-\epsilon_s)^l \right) \mathbf{1}_{\card{\mathcal{J}}} \bm{\sigma}^* + (1-\epsilon_s)^l \mathbf{I}_{\card{\mathcal{J}}}.$$
 Since $\tau_1(\mathbf{1}_{\card{\mathcal{J}}} \bm{\sigma}^*) = 0$, using the definition in \eqref{E-tau1}, it can be verified that $\tau_1(\mathbf{P}(\bm{\sigma}^*, \epsilon_s)^l) = (1 - \epsilon_s)^l$. Therefore, $$\Upsilon(\mathbf{P}(\bm{\sigma}^*, \epsilon_s)) = \sum_{l=0}^{\infty} \tau_1(\mathbf{P}(\bm{\sigma}^*, \epsilon_s)^l) = \frac{1}{\epsilon_s}.$$
From Lemma~\ref{lem:inhomo-dtmc-convergence}\ref{item:marginal-bound}, we have
\begin{align*}
& \sum_{l=0}^{T-1} \norm{\mathbf{y}(t+l) - \bm{\sigma}^*}_1	\\
 & \leq \sum_{l=0}^{T-1} \left( 2 \tau_1(\mathbf{P}(\bm{\sigma}^*, \epsilon_s)^l) + \delta_1 \Upsilon(\mathbf{P}(\bm{\sigma}^*, \epsilon_s)) \right)	\\
& \leq \frac{2 + T \delta_1}{\epsilon_s} \leq \frac{T \epsilon_g}{2 \bar{\mathsf{R}}}.
\end{align*}
Since any solution to the linear program $\mathit{L}_{\mathbf{c}^{\epsilon_p}}\left( \bm{\mu}, \bm{\lambda}+\epsilon_g \right)$ satisfies its constraints, we also have
\begin{align*}
\sum_{j \in \mathcal{J}} \sigma^*\left( \bm{\mu}, \bm{\lambda}+\epsilon_g \right)_j \sum_{h \in \mathcal{H}} \mu(h) \sum_{\mathbf{r} \in \mathcal{R}(j,h)} \alpha^*_{\mathbf{r}}(j,h) \mathbf{r} \geq \bm{\lambda} + \epsilon_g.
\end{align*}

Using the above two inequalities in \eqref{eq:service-lb} gives us
\begin{align*}
Z
 & \geq T (\bm{\lambda} + \epsilon_g/2 ),
\end{align*}
and this proves \eqref{eq:excess-service}.
\end{proof}
Combining Lemmas~\ref{lem:lyapunov-drift} and \ref{lem:fosters}, we have part~\ref{item:stab2} of Theorem~\ref{thm:algo-optimal}.

\bibliographystyle{IEEEtran}
\bibliography{energy}

\end{document}